\documentclass[12pt]{article}
\pdfoutput=1
\usepackage{amsmath,tikz,amsthm}
\usepackage{graphicx,psfrag,epsf}
\usepackage{enumerate}
\usepackage{url} 

\usepackage{hyperref,amssymb,amsmath,amsthm}
\hypersetup{colorlinks=True}
\usepackage{graphicx,xcolor,float}
\usepackage{subfig,setspace,array}

\graphicspath{{./Images/}}

\usepackage{algorithm}
\usepackage[noend]{algpseudocode}

\newtheorem{theorem}{Theorem}

\newcommand{\blind}{0}

\usepackage{multirow}

\begin{document}

\def\spacingset#1{\renewcommand{\baselinestretch}%
{#1}\small\normalsize} \spacingset{1}


\if0\blind
{
  \title{\bf Mathematics of Nested Districts: The Case of Alaska}
  \author{Sophia Caldera\thanks{Author order is alphabetical.  The authors acknowledge the Prof. Amar G. Bose Research Grant at MIT 
and the Jonathan M. Tisch College of Civic Life at Tufts University for ongoing support. 
MD is partially supported by NSF DMS-1255442.
SG is partially supported by  NSF DGE-1650441.  
SG thanks David P. Williamson and Kenrick Bjelland for helpful discussions.
We thank Coly Elhai, Mallory Harris, Claire Kelling, Samir Khan,
and Jack Snoeyink
for substantial joint work and conversations on other approaches
to modeling Alaska redistricting, and particularly acknowledge
Samir Khan for his initial implementation of prune-and-choose.
Hakeem Angulu, Ruth Buck, and Max Hully provided excellent data and technical support.
Finally, we thank Anchorage Assemblyman Forrest Dunbar for bringing this
problem to our attention. }\\
  Harvard University\\
  Daryl DeFord\\
    CSAIL, Massachusetts Institute of Technology\\
    Moon Duchin \\
    Department of Mathematics, Tufts University\\
    Samuel C. Gutekunst\\
    Cornell University\\
   Cara Nix\\
   University of Minnesota }
  \maketitle
} \fi

\if1\blind
{
  \bigskip
  \bigskip
  \bigskipd
  \begin{center}
    {\LARGE\bf Mathematics of Nested Districts: The Case of Alaska}
\end{center}
  \medskip
} \fi

\bigskip
\begin{abstract}
In eight states,  a "nesting rule" requires that each state Senate district  be exactly composed of two adjacent state House districts. In this paper we investigate the potential impacts of these nesting rules with a focus on Alaska, where Republicans have a 2/3 majority in the Senate while a Democratic-led coalition controls the House. 
Treating the current House plan as fixed
and considering all possible pairings, we find that the choice of pairings alone can create a swing of 4-5 seats out of 20 against recent voting patterns, which is similar to the range observed when using a Markov chain procedure to generate plans without the nesting constraint.
The analysis enables other insights into Alaska districting, including the 
partisan latitude available to districters with and without strong rules about nesting and contiguity.  
 \end{abstract}

\noindent%
\vfill


\newpage
\spacingset{1.45} 

\section{Introduction: Nesting}

A great deal of recent attention has been given to the problem of detecting  gerrymandering using mathematical and statistical tools. 
Much of this work has been restricted to gerrymandering in its classical form:  the manipulation of district boundaries to favor one party or another.
However, some states' rules of redistricting 
create other opportunities to extract partisan advantage from control of the process.
For example, many states favor plans that keep counties and cities intact rather 
than splitting them between districts; Iowa even requires that congressional plans keep all of its counties 
intact within districts.  Some observers worry about
whether such seemingly neutral rules would turn out to have partisan or racial consequences for
representation.  (See, for instance, \cite{VA-criteria}).  
In this paper, we will focus on a class of redistricting principles
called {\em nesting rules}, which require or encourage 
that state-level Senate districts be composed of pairs of neighboring State House or Assembly districts.

Our present case study is the state of Alaska, where 40 House districts
are paired into 20 Senate districts. We  start by focusing on the scenario in which House districts are fixed first, then subsequently paired into Senate districts. We  select two recent elections to get a baseline of partisan preference at the precinct level, then compare the current Senate plan to all others that can be formed from the current House districts by pairing.  
Across all these scenarios, we will discuss when and why the choice of pairing, or {\em perfect matching}, can have a sizeable impact on electoral outcomes.

\subsection{Perfect matching interpretation}
There are eight states that currently have
two single-member House/Assembly districts nested in each Senate district. In six of those (AL, IL, MN, MT, OR, WY),
nesting is required by State Constitution or statute, and in the 
remaining two (IA, NV), there are provisions explaining possible exemptions.
There are an additional two
states (OH, WI) that require nesting of three single-member House districts within each Senate 
district.\footnote{The article of the Ohio constitution
with this requirement was in effect in 2011 but has  now been repealed, effective 2021.}
Additionally, California, Hawaii, and New York call for nesting ``if possible.''

\begin{center}
\begin{tabular}{llll}
\small Alaska & \small 40 House $\to$ 20 Sen  & \small Illinois & \small 118 House $\to$ 59 Sen\\ 
\small Iowa & \small 100 House $\to$ 50 Sen & \small Minnesota & \small 134 House $\to$ 67 Sen \\
\small Montana & \small 100 House $\to$ 50 Sen & \small Nevada & \small 42 House $\to$ 21 Sen\\
\small  Oregon & \small 60 House $\to$ 30 Sen&\small Wyoming & \small 60 House $\to$ 30 Sen  \\
&&&\\
\small Ohio & \small 99 House $\to$ 33 Sen &
\small Wisconsin & \small 99 House $\to$ 33 Sen
\end{tabular}

\end{center}

From the perspective of election administration, nesting is convenient because it reduces the number of different ballot styles needed.  From the perspective of redistricting, nesting means that the composition of one house of the legislature massively constrains the space of possible districting plans for the other,  arguably cutting down the latitude for gerrymandering.  

When nesting is mandated, procedures can still vary.
According to the Brennan Center's {\em Citizen's Guide to Redistricting} \cite{levitt2008citizen}:
``Sometimes, a nested redistricting plan is created by drawing Senate districts first,
and dividing them in half to form Assembly districts; sometimes the Assembly
districts are drawn first, and clumped together to form Senate districts.''  This paper will focus on the second case:
matching, rather than splitting.

\subsubsection*{Proof of concept}
We begin by constructing a toy example to illustrate
that matchings matter.  
Consider the map shown in Figure \ref{fig:mapgraph1}, where each square cell represents a voter.  The 56 voters are grouped into eight equally-sized, contiguous House districts, each of which is indicated by a different color.  Geographically adjacent House districts (those overlapping on an edge) are to be paired to form four Senate districts.  
It is convenient to represent the geographic relationship of the districts with a {\em dual graph}: each {\em node} (or vertex) corresponds to a single  district, and two nodes are connected by an edge if the corresponding districts are geographically adjacent. Pairing these House districts into four Senate districts corresponds to choosing a \emph{perfect matching} in the graph: a set of four edges that, together, cover each of the eight nodes exactly once.  (See Figure~\ref{fig:mapgraph2}.)

\begin{figure}[ht]
\centering
\begin{tikzpicture}[scale=.6]
\node at (0,0) {\includegraphics[width=1.8in]{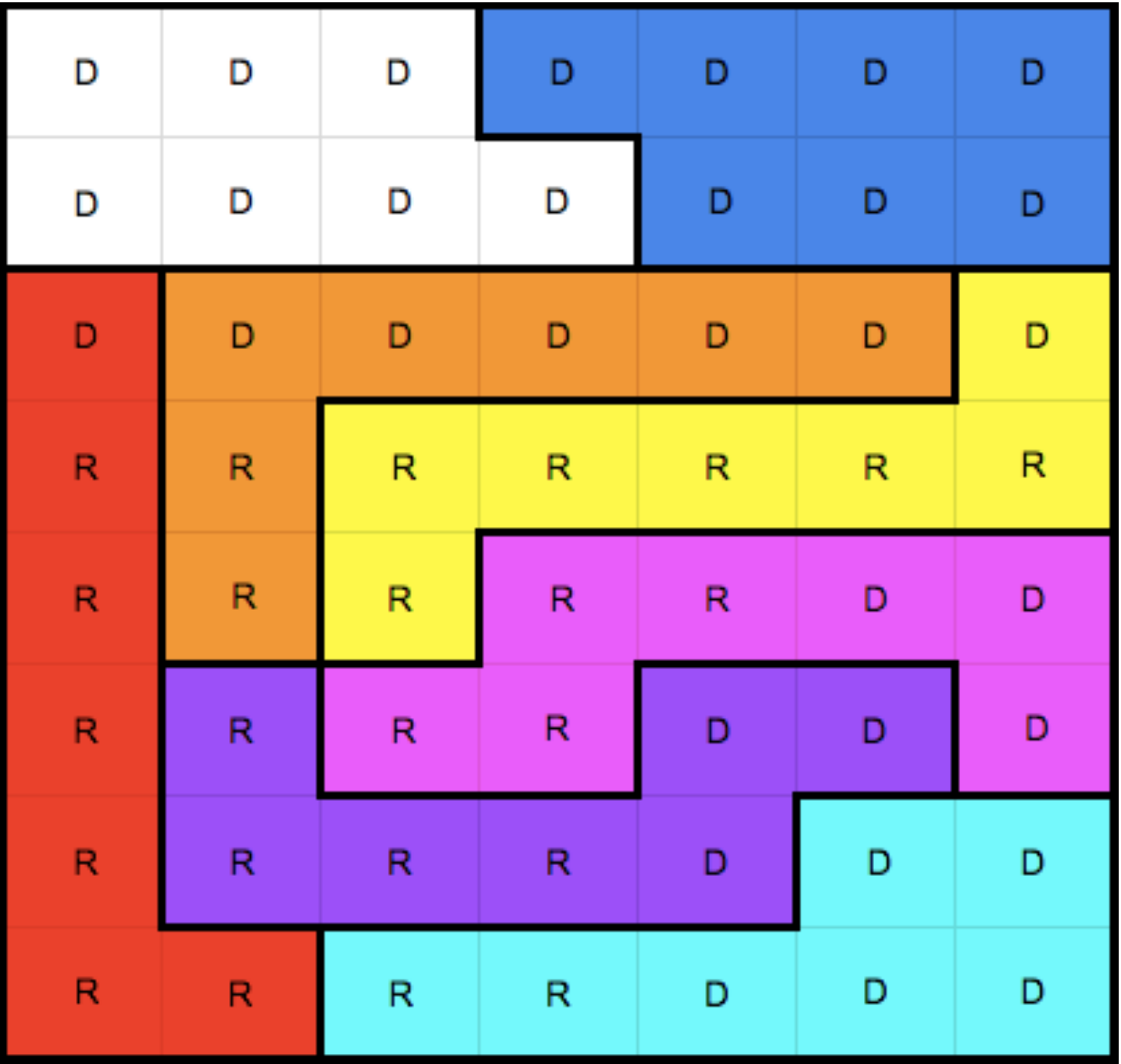}};
\begin{scope}[xshift=5.8cm,yshift=2.6cm]
    \node[circle,fill=white,draw,minimum size=1em,inner sep=2pt] (W) at (0.5,-0.5) {};
    \node[circle,fill=blue,draw,minimum size=1em,inner sep=2pt] (B) at (3.5,0) {};
    \node[circle,fill=red,draw,minimum size=1em,inner sep=2pt] (R) at (0,-4) {};
    \node[circle,fill=orange,draw,minimum size=1em,inner sep=2pt] (O) at (2,-2) {};
    \node[circle,fill=yellow,draw,minimum size=1em,inner sep=2pt] (Y) at (5,-1.5) {};
    \node[circle,fill=magenta,draw,minimum size=1em,inner sep=2pt] (M) at (5, -3)  {};
    \node[circle,fill=violet,draw,minimum size=1em,inner sep=2pt] (V) at (2.5, -3.5) {};
    \node[circle,fill=cyan,draw,minimum size=1em,inner sep=2pt] (C) at (4, -5) {};

    \draw (W) -- (B);
    \draw (W) -- (R);
    \draw (W) -- (O);
    \draw (B) -- (Y);
    \draw (B) -- (O);
    \draw (O) -- (R);
    \draw (V) -- (R);
    \draw (C) -- (R);
    \draw (O) -- (Y);    
    \draw (O) -- (V);
    \draw (Y) -- (M);
    \draw (M) -- (C);
    \draw (M) -- (V);
    \draw (C) -- (V);
\end{scope}
\end{tikzpicture}
\caption{At left, an illustrative map of 56 voters in eight equally-sized House districts to be paired into four Senate districts.  At right, the {\em dual graph} that encodes  districts  adjacency. \label{fig:mapgraph1}}
\end{figure}
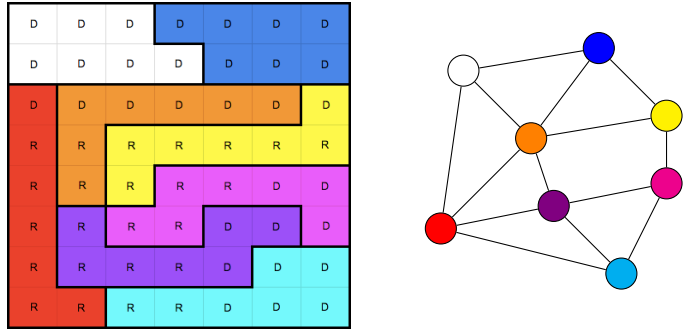

There are exactly eight perfect matchings of this graph;
in other words, 
given these House districts, there are only 8 ways to form four Senate districts while respecting nesting.
By contrast, 
there are 2,332,394,150 ways to create four 
(contiguous,  equal-size) Senate districts from these 56 units 
without that restriction.\footnote{This is the number of partitions
of a $7\times 8$ grid graph into four contiguous ``districts'' of 
14 nodes each.
See \url{mggg.org/table.html} for 
a discussion of enumeration patterns 
for districting problems on grids, and a link to enumeration code.}
If we name 
the colors White, Blue, Red, Orange, Yellow, Magenta, Cyan, and Violet, we can represent the matchings in the table below.

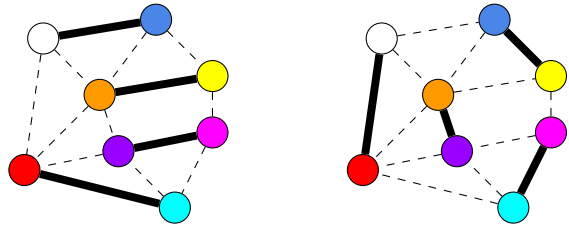
\begin{figure}[h!]
\centering
  \begin{tabular}{  | c | c | c | c | c | c | c | }
    \cline{1-3} \cline{5-7}
    \textbf{\footnotesize Matching} & \textbf{\footnotesize Results} & \textbf{\footnotesize \# D} & &
     \textbf{\footnotesize Matching} & \textbf{\footnotesize Results} & \textbf{\footnotesize \# D}     
    \\ 
    \cline{1-3} \cline{5-7}
    \scriptsize WB/RC/OY/MV & \scriptsize D/R/R/R & 1 & & \scriptsize WO/BY/RC/VM & \scriptsize D/D/R/R & 2 \\
     \scriptsize WB/RO/YM/CV & \scriptsize D/R/R/D & 2 &  & \scriptsize WO/BY/RV/CM & \scriptsize D/D/R/D & 3 \\
   \scriptsize WB/RV/CM/OY & \scriptsize D/R/D/R & 2 &  &  \scriptsize WR/BO/YM/CV & \scriptsize D/D/R/D & 3 \\ 
    \scriptsize WB/RC/OV/MY & \scriptsize D/R/D/R & 2 &  &  \scriptsize WR/BY/OV/CM & \scriptsize D/D/D/D & 4 \\
    \cline{1-3} \cline{5-7}
  \end{tabular}
  
  \vspace{.2in}

\begin{tikzpicture}[scale=.5]
\definecolor{blue}{HTML}{4a86e8}
\definecolor{red}{HTML}{FF0000}
\definecolor{orange}{HTML}{FF9900}
\definecolor{yellow}{HTML}{FFFF00}
\definecolor{magenta}{HTML}{FF00FF}
\definecolor{violet}{HTML}{9900FF}
\definecolor{cyan}{HTML}{00FFFF}
\begin{scope}
    \node[circle,fill=white,draw,minimum size=1em,inner sep=2pt] (W) at (0.5,-0.5) {};
    \node[circle,fill=blue,draw,minimum size=1em,inner sep=2pt] (B) at (3.5,0) {};
    \node[circle,fill=red,draw,minimum size=1em,inner sep=2pt] (R) at (0,-4) {};
    \node[circle,fill=orange,draw,minimum size=1em,inner sep=2pt] (O) at (2,-2) {};
    \node[circle,fill=yellow,draw,minimum size=1em,inner sep=2pt] (Y) at (5,-1.5) {};
    \node[circle,fill=magenta,draw,minimum size=1em,inner sep=2pt] (M) at (5, -3)  {};
    \node[circle,fill=violet,draw,minimum size=1em,inner sep=2pt] (V) at (2.5, -3.5) {};
    \node[circle,fill=cyan,draw,minimum size=1em,inner sep=2pt] (C) at (4, -5) {};
    \draw[line width=3] (W) -- (B);
    \draw[dashed] (W) -- (R);
    \draw[dashed] (W) -- (O);
    \draw[dashed] (B) -- (Y);
    \draw[dashed] (B) -- (O);
    \draw[dashed] (O) -- (R);
    \draw[dashed] (V) -- (R);
    \draw[line width=3] (C) -- (R);
    \draw[line width=3] (O) -- (Y);    
    \draw[dashed] (O) -- (V);
    \draw[dashed] (Y) -- (M);
    \draw[dashed] (M) -- (C);
    \draw[line width=3] (M) -- (V);
    \draw[dashed] (C) -- (V);
\end{scope}

\begin{scope}[xshift=9cm]
    \node[circle,fill=white,draw,minimum size=1em,inner sep=2pt] (W) at (0.5,-0.5) {};
    \node[circle,fill=blue,draw,minimum size=1em,inner sep=2pt] (B) at (3.5,0) {};
    \node[circle,fill=red,draw,minimum size=1em,inner sep=2pt] (R) at (0,-4) {};
    \node[circle,fill=orange,draw,minimum size=1em,inner sep=2pt] (O) at (2,-2) {};
    \node[circle,fill=yellow,draw,minimum size=1em,inner sep=2pt] (Y) at (5,-1.5) {};
    \node[circle,fill=magenta,draw,minimum size=1em,inner sep=2pt] (M) at (5, -3)  {};
    \node[circle,fill=violet,draw,minimum size=1em,inner sep=2pt] (V) at (2.5, -3.5) {};
    \node[circle,fill=cyan,draw,minimum size=1em,inner sep=2pt] (C) at (4, -5) {};

    \draw[line width=3] (W) -- (R);
    \draw[dashed] (W) -- (O);
    \draw[dashed] (W) -- (B);
    \draw[line width=3] (B) -- (Y);
    \draw[dashed] (B) -- (O);
    \draw[dashed] (O) -- (R);
    \draw[dashed] (V) -- (R);
    \draw[dashed] (C) -- (R);
    \draw[dashed] (O) -- (Y);    
    \draw[line width=3] (O) -- (V);
    \draw[dashed] (Y) -- (M);
    \draw[line width=3] (M) -- (C);
    \draw[dashed] (M) -- (V);
    \draw[dashed] (C) -- (V);
\end{scope}
\end{tikzpicture}
\caption{The districts can be matched the eight different ways listed here, leading to the Democratic party getting anywhere from 25\% to 100\% of the Senate seats.
The two perfect matchings corresponding to the extreme outcomes are shown here.\label{fig:mapgraph2}}
\end{figure}

In this toy example, we discover that 
the choice of matching can swing the outcome for Democrats from 1 seat to 4 seats out of four.
Below, we carry out a similar analysis on real-world data.

\subsection{Mathematical literature on perfect matchings}\label{sec:lit}

In our motivating example, we considered the Senate outcomes for every possible perfect matching in a small graph.  Enumerating all perfect matchings in a given graph is a classical problem in the mathematical field of combinatorics; 
it has captured significant 
attention because it is at once quite elementary and extremely difficult to compute 
for arbitrary graphs \cite{valiant}.
The matching problem is also of great interest to 
physicists studying dimer coverings (domino tilings) of lattices, which are used to estimate thermodynamic
behavior of liquids \cite{kenyon-okounkov}.
In 1961, three statistical physicists,  Temperley, Fisher, and  Kasteleyn, independently and nearly
simultaneously derived the formula for the number of perfect
matchings of an $m\times n$ grid  \cite{kasteleyn_statistics_1961,temperley_dimer_1961} and subsequently proposed the FKT algorithm
for efficiently computing the
number of perfect matchings of any {\em planar} graph (that is, in any graph that
can be drawn in the plane without  edges crossing).
The algorithm is discussed in more detail
in Appendix B in the supplement.
For surveys on the mathematics of matching, see \cite{Lov09} and Volume A of \cite{Sch03}.%
\footnote{This problem is also intimately related to a second
enumeration problem, that of counting the
{\em spanning trees} of a graph.  The number of spanning
trees is sometimes called the {\em complexity} of a graph.
Temperley defined a transformation that starts with
a graph and creates a new associated graph called its
$T$-graph.
A series of remarkable theorems tell us that if $G$ is
the $T$-graph associated to $\overline G$, then
the number of spanning trees of $\overline G$
is {\em exactly}
equal to the number of perfect matchings of $G$ \cite{burton-pemantle,kenyon,kpw,temperley}.}

\subsection{Paper outline}\label{sec:outline}

The central research question here is to quantify the 
partisan advantage available to an agent who is empowered only 
to select a House-to-Senate pairing.
In Alaska, where there are only 40 House districts which 
are patterned in a not very dense manner, it might seem that 
there is only limited advantage to be gained.  However, we will demonstrate
that the choice of pairings alone can create a swing of 4-5 seats out of 20 against recent voting patterns.  
In fact, we will see that even though pairings give a far simpler model of how to create Senate districts,
they give just as much partisan latitude as making Senate districts from scratch.  

We begin by reviewing pertinent background on Alaska politics, demographics, and  redistricting rules in \S\ref{sec:AK}, culminating in the selection of two recent elections---the Governor and U.S. House races of 2018---to serve as our electoral baselines for the remainder of the analysis.  
In \S\ref{sec:methods}, we begin by describing the construction of dual graphs that model the adjacencies of geographical 
units-- in this case, House districts. Next, we overview the algorithmic approaches we apply to those graphs in the rest of the
paper. These methods include enumerating matchings with a classic algorithm called FKT, constructing sets of matchings with a depth-first algorithm we call {\em prune-and-choose} described in Appendix C, and finally varying the underlying districts with a Markov chain.  The proof of validity for prune-and-choose
is found in Appendix C.2.

In the remainder of the paper, we report on the results of these
algorithmic investigations for Alaska. Beyond the flexibility inherent in choosing the nesting, 
we find that the interpretation of redistricting rules (in particular, geographic adjacency when regions are connected by water) has a substantial impact on the number of matchings.  With the current House districts fixed, \S\ref{sec:altmatch} measures the partisan tilt of the pairing itself among the full set of matchings.
Finally, in \S\ref{sec:chain}, we vary the House and Senate districts themselves by randomly assembling them
from precinct building blocks with a method that provides heuristic assurances of representative sampling.
By exploring the space of valid plans, and evaluate expected partisan properties and matchability of alternative plans.

 In Appendix D we extend this analysis by enumerating the perfect matchings
in each of the eight states that mandate two-to-one nesting.
For several of these states, it would be computationally infeasible
to construct the complete set of matchings because it is prohibitively large;
nonetheless,  Appendix E  describes how they can 
 be sampled efficiently.

\section{Alaska electoral politics}\label{sec:AK}

\subsection{Partisanship in Alaska}

Alaska is an outlier in U.S. political geography
for several reasons including its
uniquely wide array of viable minor parties featured
in both local and statewide races.
For example, the state officially 
recognizes the secessionist Alaskan Independence Party, which 
succeeded in electing Wally Hickel as Governor in 1990.
In addition, there are nine organized ``political groups''
that are seeking official recognition, and meanwhile are entitled
to run candidates for statewide office:  the Libertarian,
Constitution,
Progressive, Moderate, Green, and Veterans Parties, together with 
the more fringe OWL Party, Patriot's Party, and 
UCES Clowns Party.\footnote{See  \url{http://www.elections.alaska.gov/Core/politicalgroups.php}.}

The current Governor of Alaska is 
Republican Mike Dunleavy, whose predecessor
Bill Walker won as an 
Independent in 2014, becoming the only sitting U.S. Governor 
not from one of the
two major parties at the time. (Walker had previously left the Republican Party and then successfully ran as an Independent candidate, with a Democratic candidate for Lieutenant Governor.) In 2014, Alaska had the first U.S. Senator in more than 50 years to win election as a write-in candidate, Senator Lisa Murkowski.

\begin{figure}[!ht]
\centering
\includegraphics[height=1.6in]{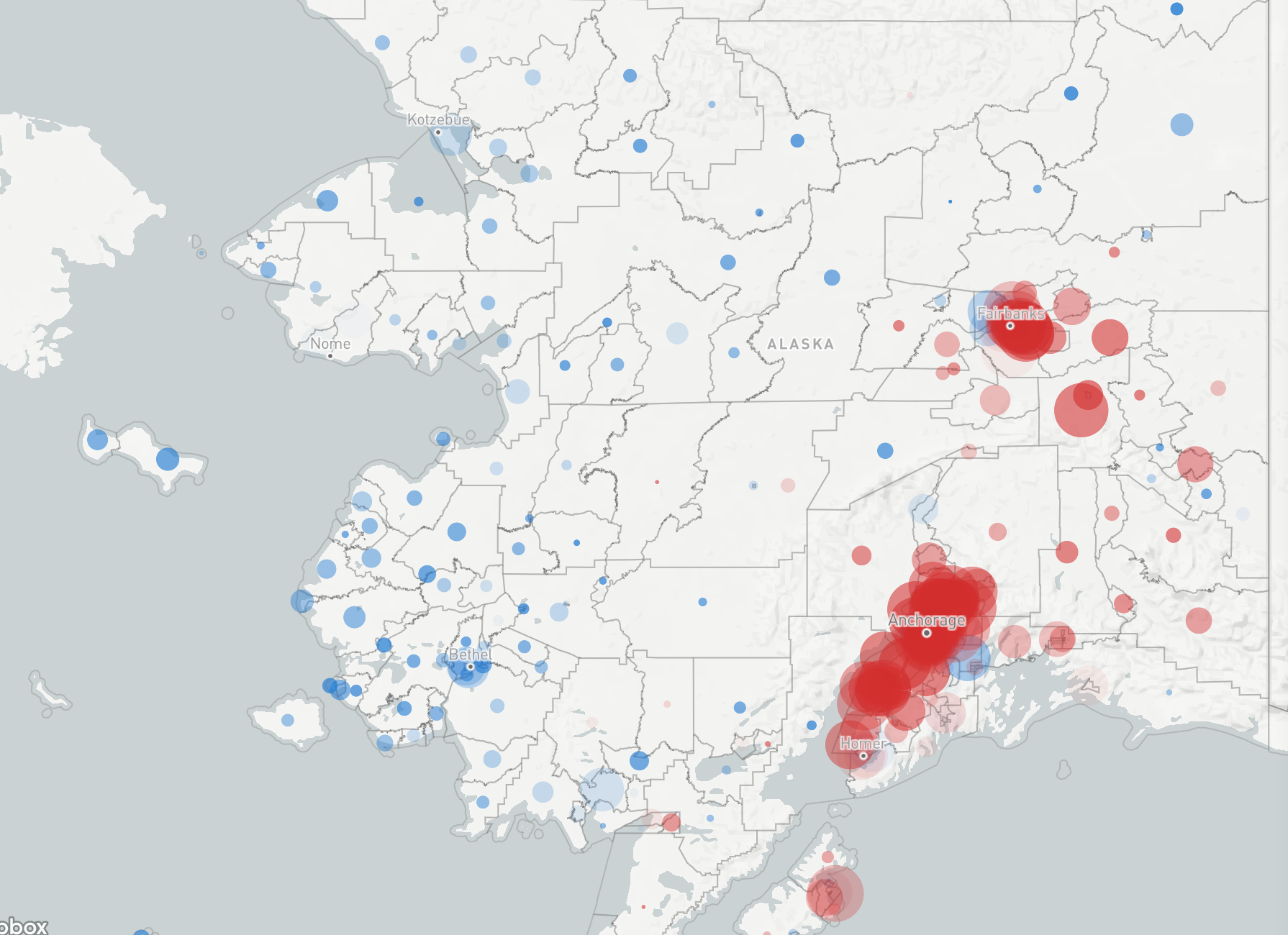}\quad 
\includegraphics[height=1.6in]{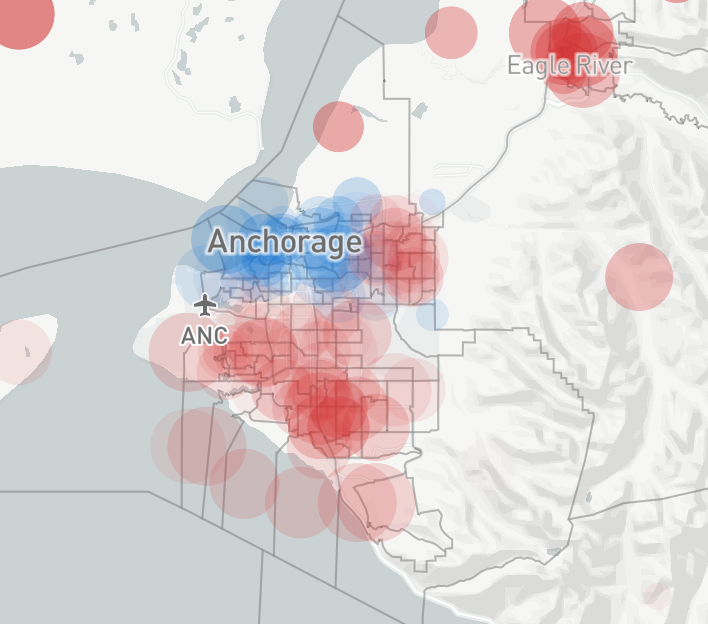}
\caption{Trump share of major-party Presidential
vote per precinct, sized by number of votes cast.}
\end{figure}

Alaska is also unique in its geographic distribution of the major parties' strengths. Unlike the contiguous United States, where urban areas tend to be most reliably Democratic, Alaska has Democratic strength in the rural areas to the north and west of the state. These areas are the homes of significant numbers of Native Alaskan residents, who constitutes the largest minority in the state. Conversely, the Republican vote is often stronget in suburban areas.
In fact, even the city of Anchorage---by far the most
populous in the state with 291,826 out of
Alaska's 710,231 residents in Census 2010---votes
Republican overall in recent presidential
races, making it a rare city of its size
to do so.\footnote{In the 2016 Presidential race, Trump's share of the major party 
vote was 58.4\% statewide and 
53.1\% in Anchorage.  The trend 
holds up across elections in the last cycle,
with Republican performance in Anchorage trailing statewide levels by only
about four points.}

Although it is one of the ``reddest'' states in 
national terms, the Republican-Democratic split is not the fundamental divide in Alaskan politics. 
Extremely conservative Republicans are sometimes 
balanced by a tenuous coalition of  moderate Republicans, Democrats, and Independents, which 
currently aligns to give net Democratic control
in the state House.
In 2018, an Independent, Libertarian, or Nonpartisan candidate ran in nine of the 40 House districts; an Independent won in one district and one Democratic candidate changed his affiliation to undeclared after winning \cite{BP}. In areas where the Democratic party label is an obstacle to election, running as an Independent can be a successful political strategy. The majority caucus in the House originally consisted of 25 members: all 15 Democrats, the two unaffiliated members, and eight  Republicans \cite{BP}.\footnote{Twenty-one members (15 Democratic, 4 Republican, and two unaffiliated) voted together to elect the current Speaker (who ran for his House seat as a Democrat but became unaffiliated just days before being elected Speaker).  Four more Republicans joined to establish the majority caucus.  In May 2019, however, one Republican left the House majority coalition \cite{BP, AP}.} On the other hand, one state Senator elected as a Democrat caucuses with the Republican majority in that body \cite{hoffman}.

\subsection{Racial demographics and the Voting Rights Act}

The 2010 Census reports Alaska's racial demographics as roughly 6\% Hispanic,
with non-Hispanic population comprising
63\% White, 3.5\% Black, 
5.5\% Asian, and 15\% Alaska Native or other Native American as shares of the total. 
An additional 9\% of residents are recorded as belonging to other races,
or to two or more races.  Figure 
\ref{fig:native} shows the proportion of Alaska Native or other Native American residents across the state.

\begin{figure}[h]
    \centering
    \includegraphics[height=1.6in]{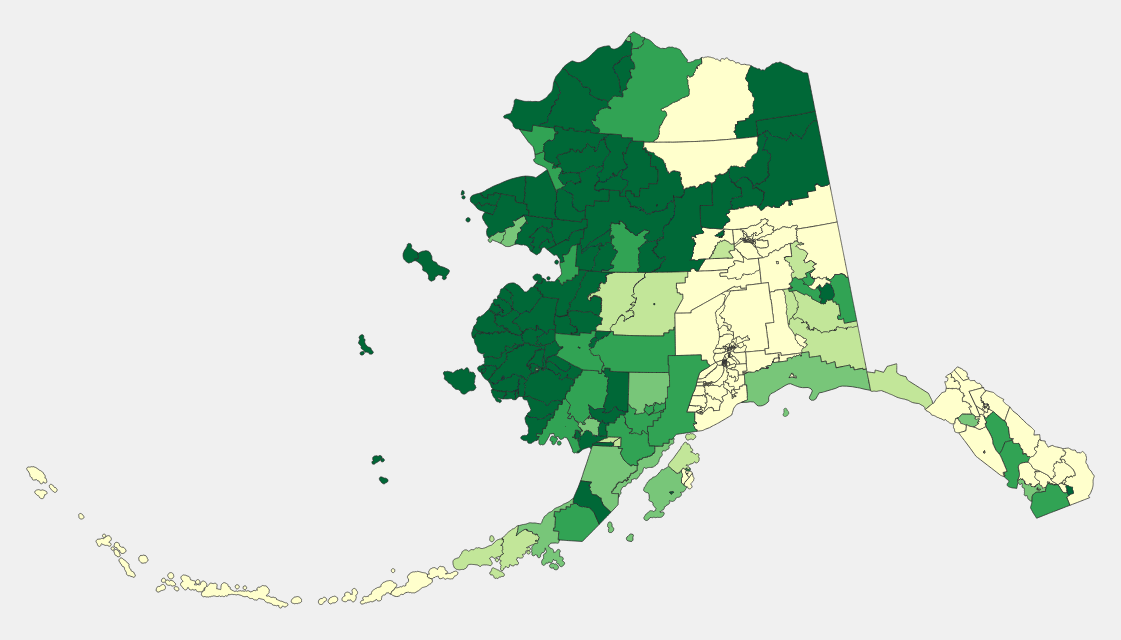}
\caption{Proportion of Alaska Native or other Native American residents across Alaska.
The color scale is in equal intervals of
20\%; the darkest shade marks
precincts that are 80-100\% 
Native.\label{fig:native}}
\end{figure}

The large Alaska Native population
has long been singled out for federal protection under the  Voting Rights Act of 1965, specifically
Section 5 of the VRA, which  required covered  jurisdictions to seek prior federal approval
(or ``preclearance'') for any changes to districts or other voting laws.
Alaska's inclusion owed to a long history
of discriminatory ``literacy tests''---in this case, English-language tests used to deny
voting eligibility to Native residents---making Alaska one of only nine states covered in full by the special
protections \cite{DOJ}.\footnote{Since 1971, the indigenous people of Alaska 
are organized into thirteen regional Tribal Corporations to administer land and finances.  The current legal landscape
gives the corporations substantial financial clout, which does not translate to commensurate political representation
for the broader Native Alaskan population.}
Though the Supreme Court ended
the practice of preclearance with
Shelby v. Holder (2013), all states are still bound by the 
VRA requirement to afford minority groups the ability to elect a candidate of their choice where possible.\footnote{State courts have established that  Alaska
redistricters must consider the state's constitutional requirements for districts before considering the requirements of the Voting Rights Act. Alaska's courts have enforced this hierarchy several times, including most recently in a 2012 ruling that invalidated the maps used in that year's elections \cite{Mauer_ADN2013}.}
Issues of fair representation and ballot access for the rural
Native population are still highly active in Alaska \cite{caldwell}.

\subsection{Redistricting rules and practices}\label{AK-rules}

Following a 1998 state constitutional amendment, a five-member Alaska Redistricting
Board  was formed to draw new district lines after each decennial census \cite{Epler_ADN2011}. The House speaker, Senate president, and Chief Justice of the state Supreme Court each choose one member of the board, and the Governor chooses two.
At least three members of the board must approve a redistricting plan for it to be adopted. The board must draw maps in accordance with the state Constitution, which requires that House districts be ``contiguous and compact territory containing as nearly as practicable a relatively integrated socio-economic area.... {[}and{]} contain{[}ing{]} a population as near as practicable to the quotient obtained by dividing the population of the state by forty'' while Senate districts are simply  ``...composed as near as practicable of two contiguous house districts'' without further constraints \cite{constitution}.\footnote{As far as we are
aware, the socio-economic clause has never been operationalized or enforced.}

Balancing the  requirements of the VRA and the  guidelines of the state Constitution---compactness in particular---means that Alaska's House and Senate districts have to be drawn in a coordinated fashion in most of the state.
However, Alaska's relatively urban centers of Anchorage and Fairbanks are both 
predominantly white and made up of small, regular pieces. This homogeneity of demographics and geography provides additional flexibility in these regions for the map drawer to construct House districts first, without considering potential Senate pairings. 

Allegations of partisan intent have frequently been leveled at the redistricting process in Alaska. The maps drawn after the 2000 Census were accused of being a Democratic gerrymander, while Democrats have called the post-2010 maps (drawn by a board with a 4-1 Republican majority) a Republican gerrymander \cite{Mauer_ADN2013}. 
The fact that a Democratic-led caucus controls the House while Republicans have 2/3 control 
of the Senate lends credence to the possibility that not the House
districts themselves, but their pairing to form Senate districts, is chosen for Republican advantage.  
That possibility is investigated below.

\subsection{Our choice of election data}

In Alaska, three types of races occur statewide. The entire state votes for a Governor and Lieutenant Governor, elected on a single ticket, every four years; they elect one member to the U.S. House of Representatives every two years; and they elect a U.S. Senator for a term of 6 years in the Class 2 and Class 3 cycles.

We consider only those elections which occurred after the implementation of new maps in July 2013. (A map approved for temporary use in 2012 was replaced after litigation.) 
Seven statewide races occurred in this time period:  

{\footnotesize
\begin{tabular}{rlcll}
Gov14& {\bf Walker} (I/D) & [48.10\%] &
Parnell (R) &[45.88\%] \\
Cong14& {\bf Young} (R) &[50.97\%] & Dunbar (D) &[40.97\%]\\
&&& McDermott (L) &[7.61\%]  \\
Sen14& {\bf Sullivan} (R) &[47.96\%] &
 Begich (D) &[45.83\%]  \\
Cong16& {\bf Young} (R)&[50.32\%] &
Lindbeck (D) &[36.02\%]\\  
&&& McDermott (L) &[10.31\%] \\
Sen16& {\bf Murkowski} (R) &[44.36\%] & 
Miller (L) & [29.16\%]\\ 
&&& Stock (I) &[13.23\%]\\
&&&Metcalfe (D) &[11.62\%] \\
Cong18& {\bf Young} (R) & [53.08\%] &
Galvin (D) & [46.50\%]  \\
Gov18& {\bf Dunleavy}
(R) & [51.44\%] & Begich (D)& [44.41\%] 
\end{tabular}
}

\noindent The list includes all candidates with at least 5\% of the vote in any race.%
\footnote{Walker ran for Gov14 as an Independent, but with a Democratic running mate.  In Gov18, Walker dropped
out and ultimately received just 2\% of the vote.}

\begin{figure}[!ht]
    \centering
   \begin{tikzpicture} [scale=.8]
\node at (0,0)   { \includegraphics[width=2.2in]{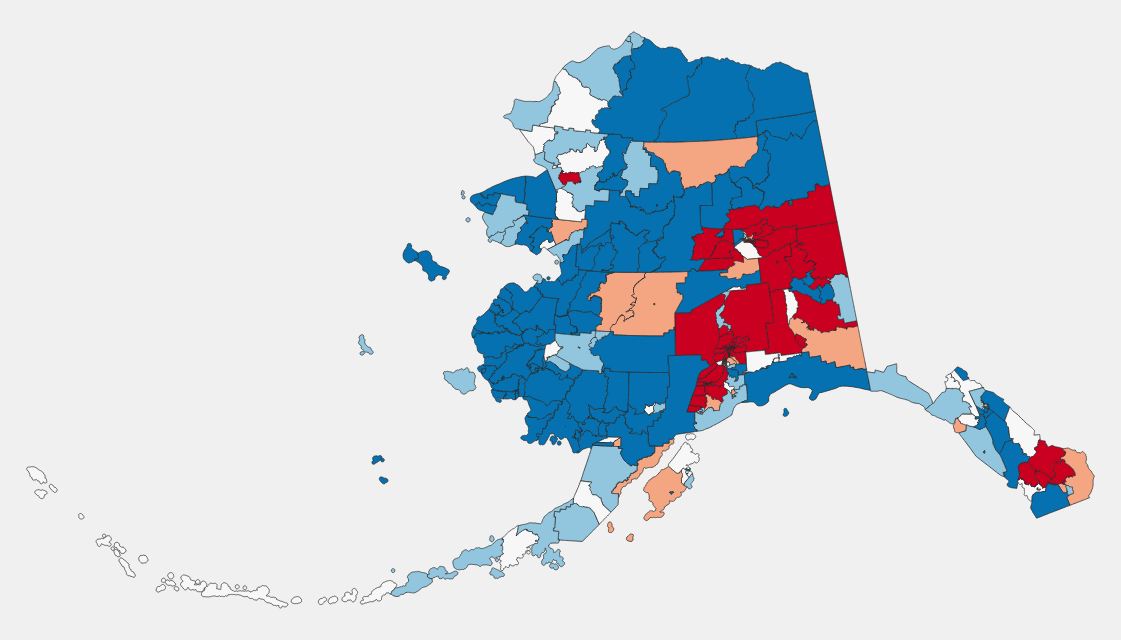}};
\node at (7.5,0)   { \includegraphics[width=2.2in]{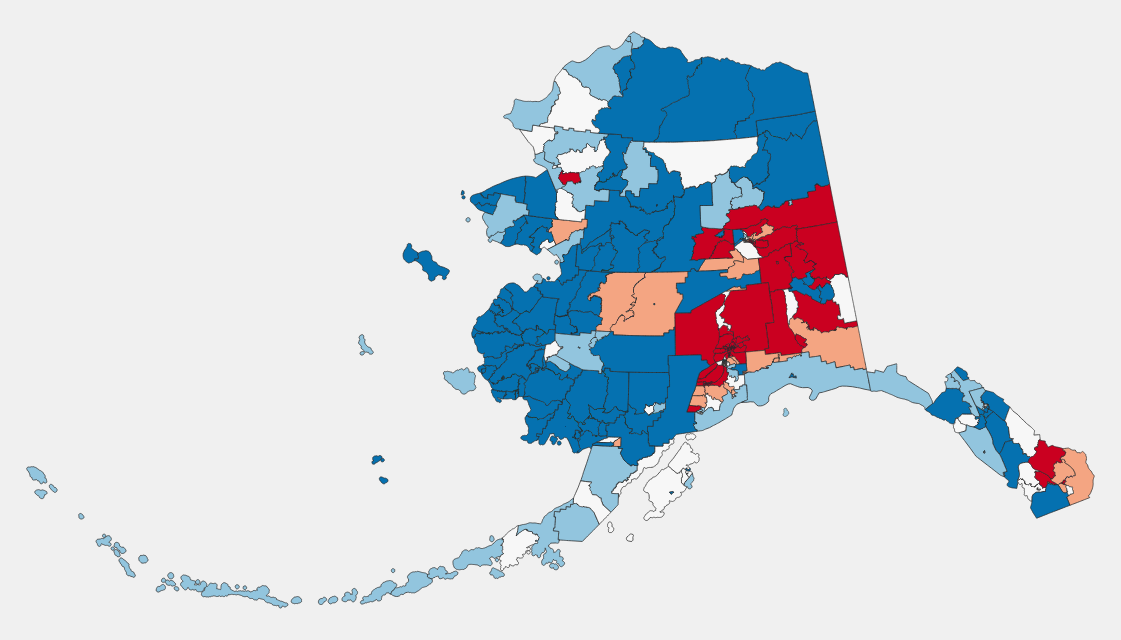}};
\node at (0,5)   { \includegraphics[width=2.2in]{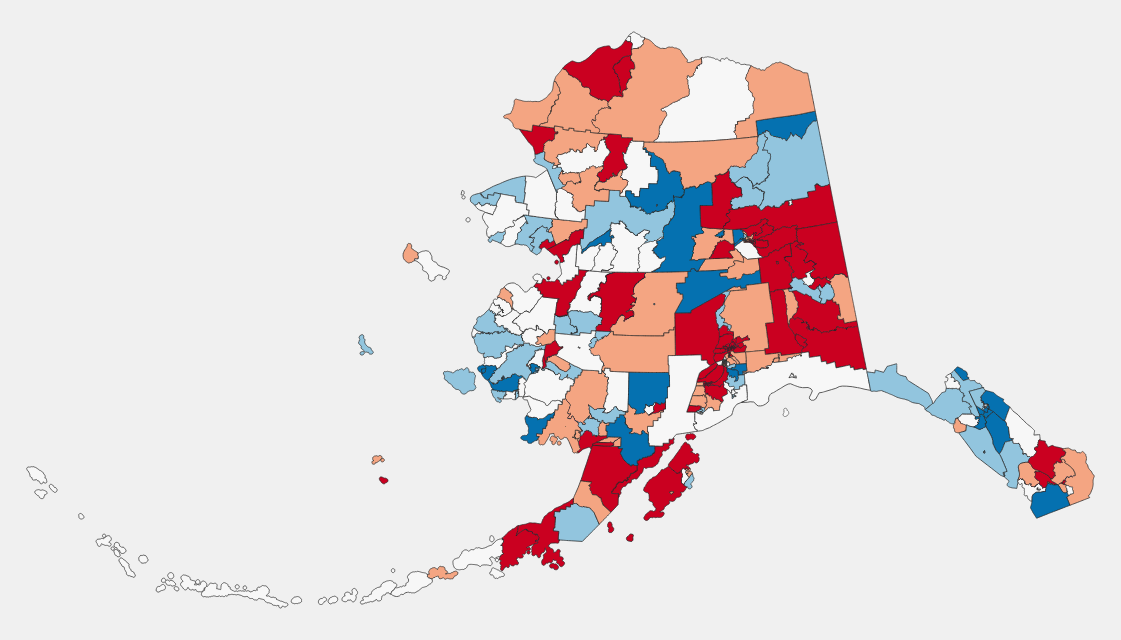}};
\node at (7.5,5)   { \includegraphics[width=2.2in]{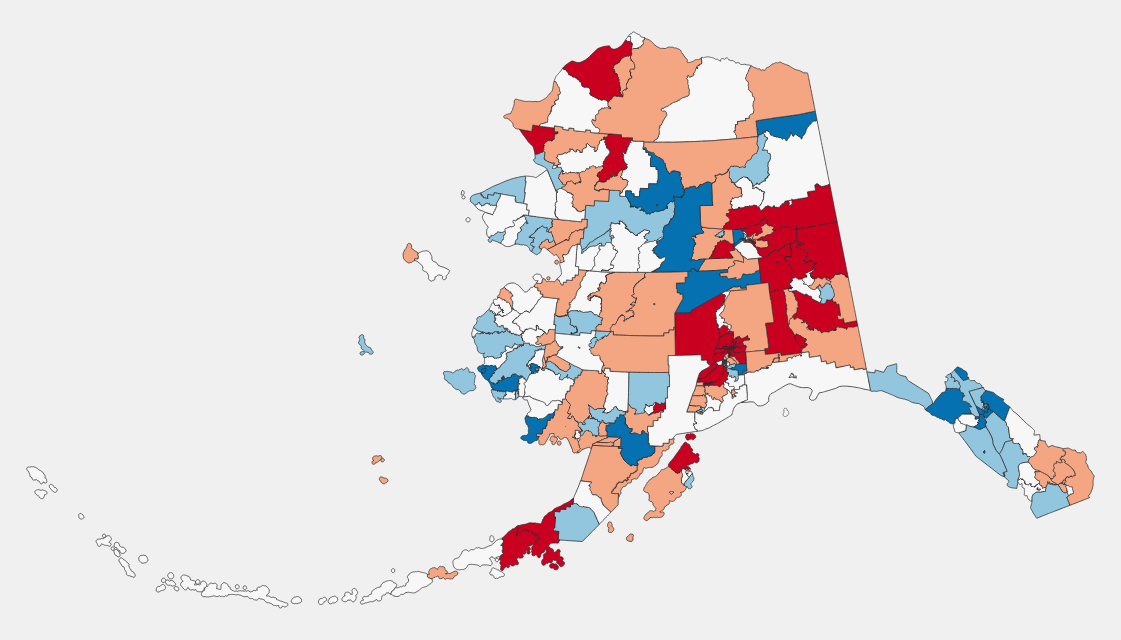}};
	\node at (0,-2.5) {\small Gov18-N party share};
	\node at (7.5,-2.5) {\small Gov18-A party share};
	\node at (0,2.5) {\small Cong18-N party share};
	\node at (7.5,2.5) {\small Cong18-A party share};
    \end{tikzpicture}
\caption{These choropleth images show that 
 party preferences in the Governor race are spatialized very differently from the U.S. House race,
even though the statewide party share is nearly identical.
On the other hand, there is  little visible change with and without including
absentee ballots (marked with A and N, respectively), though this does have a significant bottom-line partisan impact.}
\end{figure}
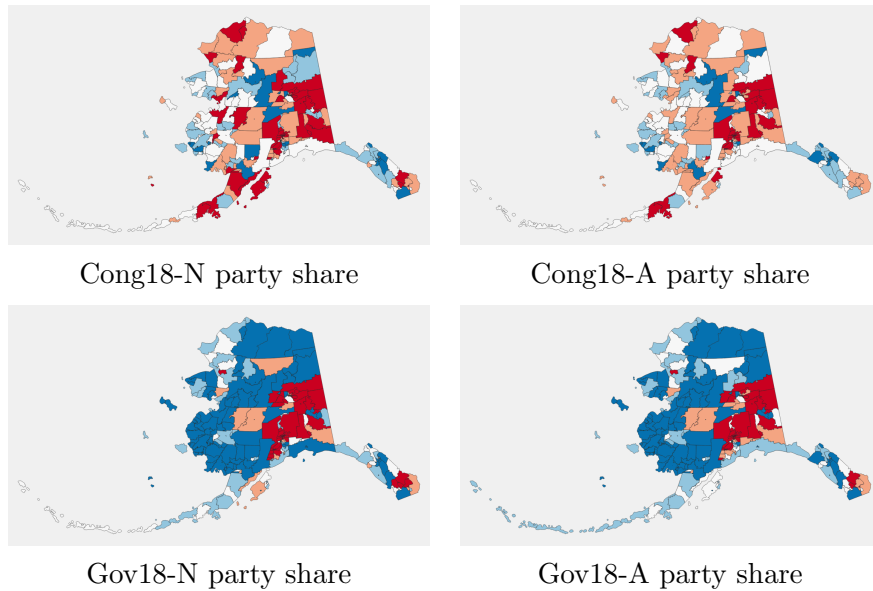


We will use the Cong18 and Gov18 races as the fundamental electoral data for the analysis below. 
These two  contests feature a  Democratic and  Republican candidate without major third-party presence
and are interesting because they have very 
different spatial patterns of party support but
similar bottom-line partisan shares.
The two-party vote share for the Democrat in 
those races was 46.7\% for Alyse Galvin against Don Young
for Congress and 46.3\% for Mark Begich against
Mike Dunleavy for Governor.%
\footnote{ We note that the question of preferring endogenous or exogenous election data for redistricting analysis is a live one in political science, as reflected for instance in the article, rejoinder, and response between Best et al and McGhee in the March 2018 issue of the Election Law Journal \cite{best_considering_2017,mcghee_rejoinder_2017,best_authors_2017}. Our research group inclines to the use of well-chosen exogenous (statewide) election
data in general, but we further note that using endogenous data would be forbidding in the Alaska legislature. Besides a significant number of uncontested races, these legislative races also feature a proliferation of minor parties (described above), making a regression analysis particularly inadequate to cleanly model voters' preferences between the two major parties.}

Alaskan elections receive a high proportion of ballots not reported through individual precincts. These unprecincted ballots include absentee, provisional, and early votes, which are all reported by legislative
House district.
For federal races such as the U.S. House, a small
number of overseas military ballots are also 
reported on a statewide basis. 
In the 2018 U.S. House race, 33.36\% of reported ballots were unprecincted. In the 2018 Governor race, 34.18\% of ballots were unprecincted. 
(For ease of reference, we will call all unprecincted
ballots  ``absentee''  below.)
We report results for each election both including and excluding the absentee ballots. 
Thus our four {\em election treatments} can be 
labeled Gov18-N, Gov18-A, Cong18-N, and Cong18-A,
where the N versions drop absentee ballots from the tally and the A versions include them.
In \S\ref{sec:chain}, we need to know 
the precinct location of the votes;
for this, we assign absentee ballots
to precincts in numbers proportional to 
precinct population.

 Table \ref{tab:elecdata}
shows clearly that absentee ballots collectively favor Democrats
by roughly two percentage points per House
district.  Only districts 37 and 38 have 
a Republican lean in the absentee ballots in 
either election.

\begin{table}[ht]
\centering
\includegraphics[width=4.5in]{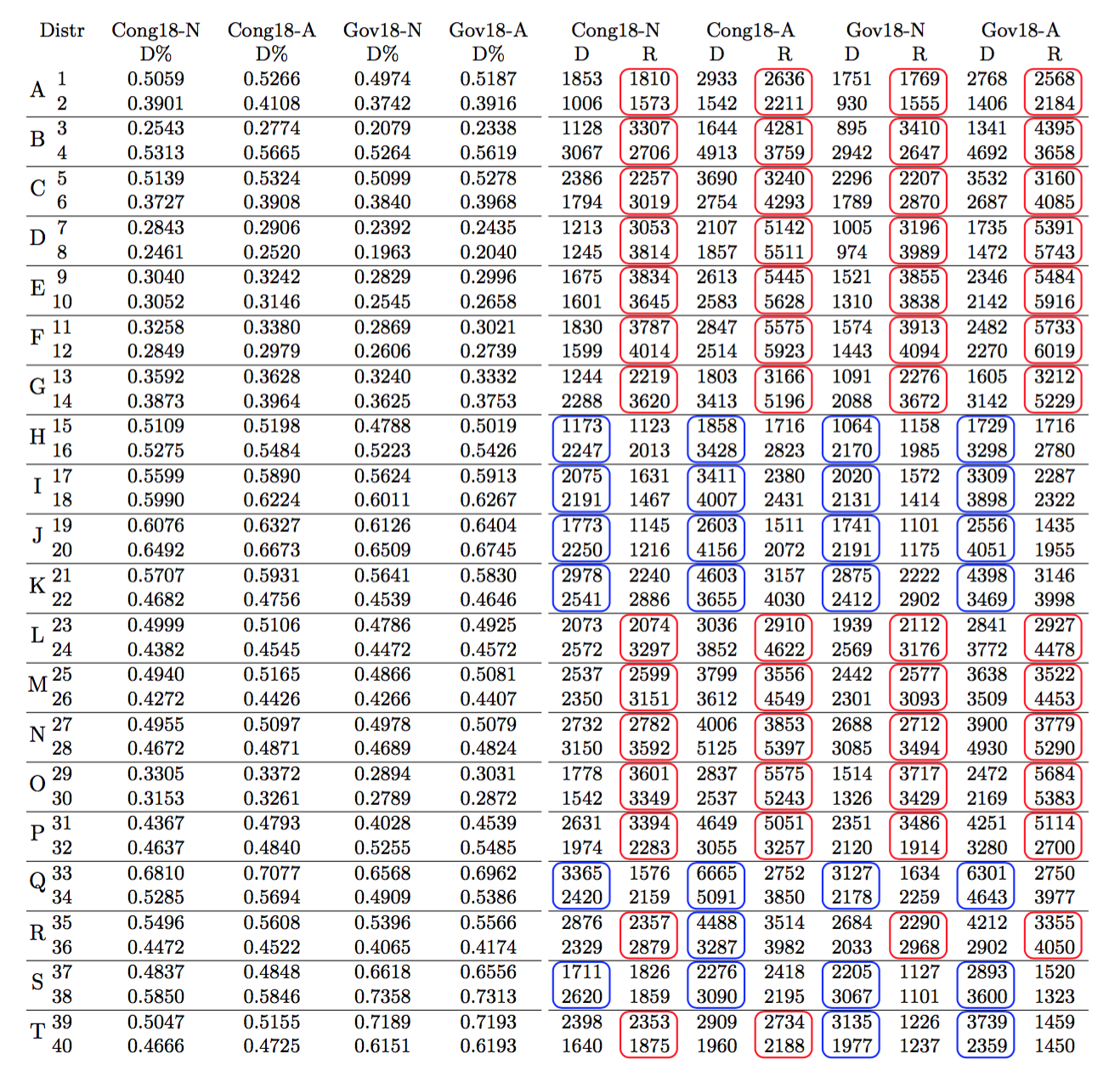}
\caption{This table shows how the currently enacted House and Senate districts fall with respect to the voting patterns from our two selected elections. The individual rows correspond to the 40 numbered state House districts, paired into the 20 state Senate districts labeled A-T. The highlighting in the votes columns shows which party  received the majority of the two-way vote share in the corresponding Senate district.}\label{tab:elecdata}
\end{table}

\newpage
\section{Data and methods}\label{sec:methods}

All experiments in this paper were performed on an Ubuntu 16.04 machine with 64 GB memory and an Intel Xeon Gold 6136 CPU (3.00GHz). Algorithmic descriptions of the FKT method for enumerating matchings and the Prune-and-Choose method for generating matchings are given in Appendix B and C of the Supplementary Material, respectively. 

\subsection{Election results, shapefiles, and dual graphs}
Election data were gathered from the \href{http://www.elections.alaska.gov/doc/info/ElectionResults.php}{Alaska elections website}  \cite{AK_Elections} and demographic data were obtained from the 2010 Census. Absentee and early voting information was only available by House district, so precinct-level data was assigned by prorating by population. 
We prepared the geospatial data with 
the \href{https://github.com/mggg/maup}{MGGG Preprocessing Suite}, which uses 
areal interpolation for blocks not fully
contained in precincts \cite{maup}.  The cleaned and processed version of the data is available on GitHub \cite{mggg-states}.

In the 2010 Census, Alaska had 45,292 census blocks, of which  over a third (18,263) are water-only.  
Alaska has 441 precincts, ranging in population from a minimum of 44 people (Pedro Bay) to 7994 people (JBER1, in Anchorage) \cite{AK_Elections}.  
Six precincts have over 5000 people, and sixteen have under 100.  

Beginning with a shapefile of the geography, in this instance precincts, we use geospatial libraries in Python 
to create a dual graph in {\sf GerryChain} whose nodes are the geographic units, and where two units are connected by an 
edge if their units share a positive-length boundary in the shapefile \cite{gerrychain}.  We then adjust edges to better
correspond to plausible notions of adjacency, especially when water is involved, as described below.

\subsubsection*{Water adjacency}

For areas connected only by water, a decision must be made about whether to count them as adjacent. 
To illustrate the impact of this seemingly
minor issue, we construct three different dual graphs of the 
precinct map, which we call the {\em tight},
{\em restricted}, and the 
{\em permissive} graphs.  

Permissive adjacency is the closest match to the AK Division of Elections precinct shapefile.  The dual graph of those precincts is nearly connected using this approach, except for one gap 
in the Kodiak archipelago and five additional island precincts 
of the West coast.  We manually added all visually reasonable connections in these cases.  Among the 441 precincts, this process
produces 1151 edges.  Aggregating the 
precincts into the 40 current House districts
produces a House district dual graph with 100 edges. 

\begin{figure}[!ht]
\centering
\includegraphics[height=1.6in]{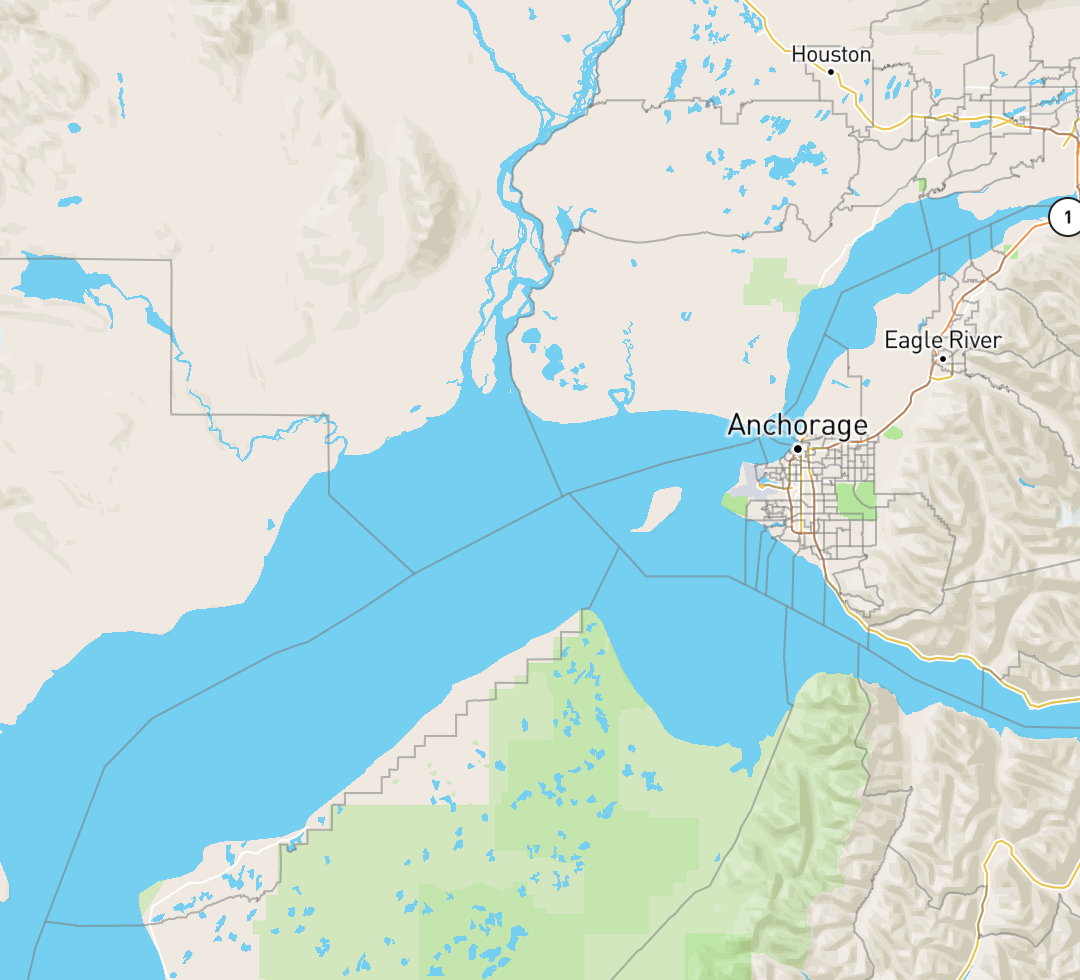} \quad 
\includegraphics[height=1.6in]{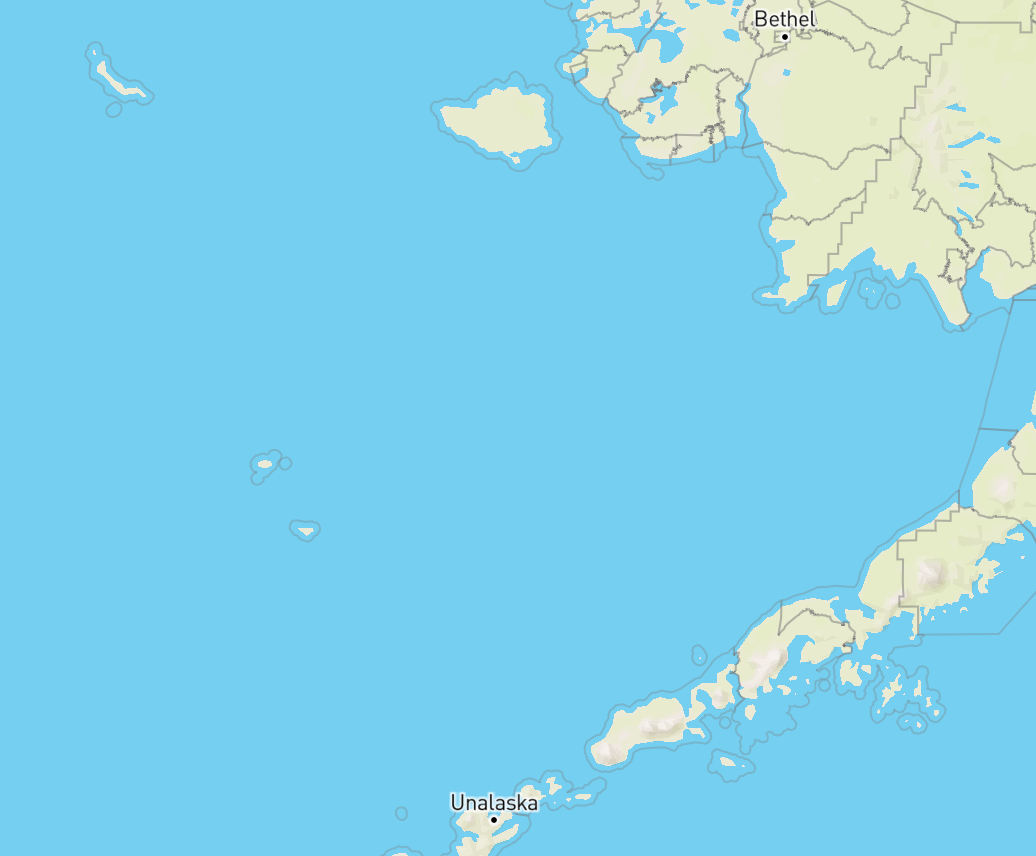}
\caption{The Cook Inlet is a body of water
stretching up from the Gulf of Alaska;
its Knik Arm and Turnagain Arm  surround
the densest part of Anchorage, separating it from
rural precincts to the north and south.
Following precinct adjacencies provided by the 
state would allow districts to jump across the water,
while a more restrictive notion of adjacency
would not.  On the right, we see that the precinct
shapefile gives no guidance on how the islands
are allowed to be connected to the mainland by 
districts.\label{fig:water}}
\end{figure}

To construct our more restricted 
notion of adjacency, we consulted
the Census Bureau Cartographic 
Boundary shapefile, which is {\em clipped to land}, i.e.,
excludes water from its geographic units.
With this as a guide, we removed certain connections across
water (see Figure~\ref{fig:water}).
This reduces the number of edges
modestly, from 1151 to 1109 for the precinct dual graph
and from 100 to 92 for the House dual graph. 

Finally, we create the tightest version of the graph by using the current House map as a guide, keeping the fewest water adjacencies that would allow the 
current districts to be considered valid. This gives us a tight dual graph with 1105 precinct edges and 89 House edges. As we will see below, these small changes to the underlying dual graph can have large consequences for the number of possible matchings. Figure \ref{fig:AKDG} shows the resulting graphs, which we use in the remainder of the analysis.

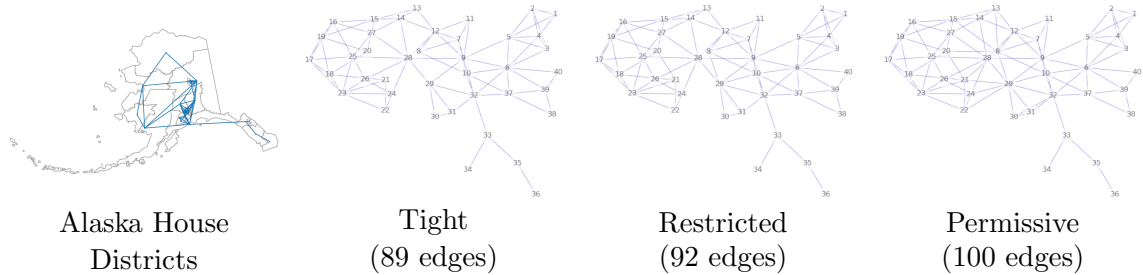
\begin{figure}[!h]
    \centering
\begin{tikzpicture}    [scale=.7]
\node at (-11,0) {\includegraphics[height=1.1in]{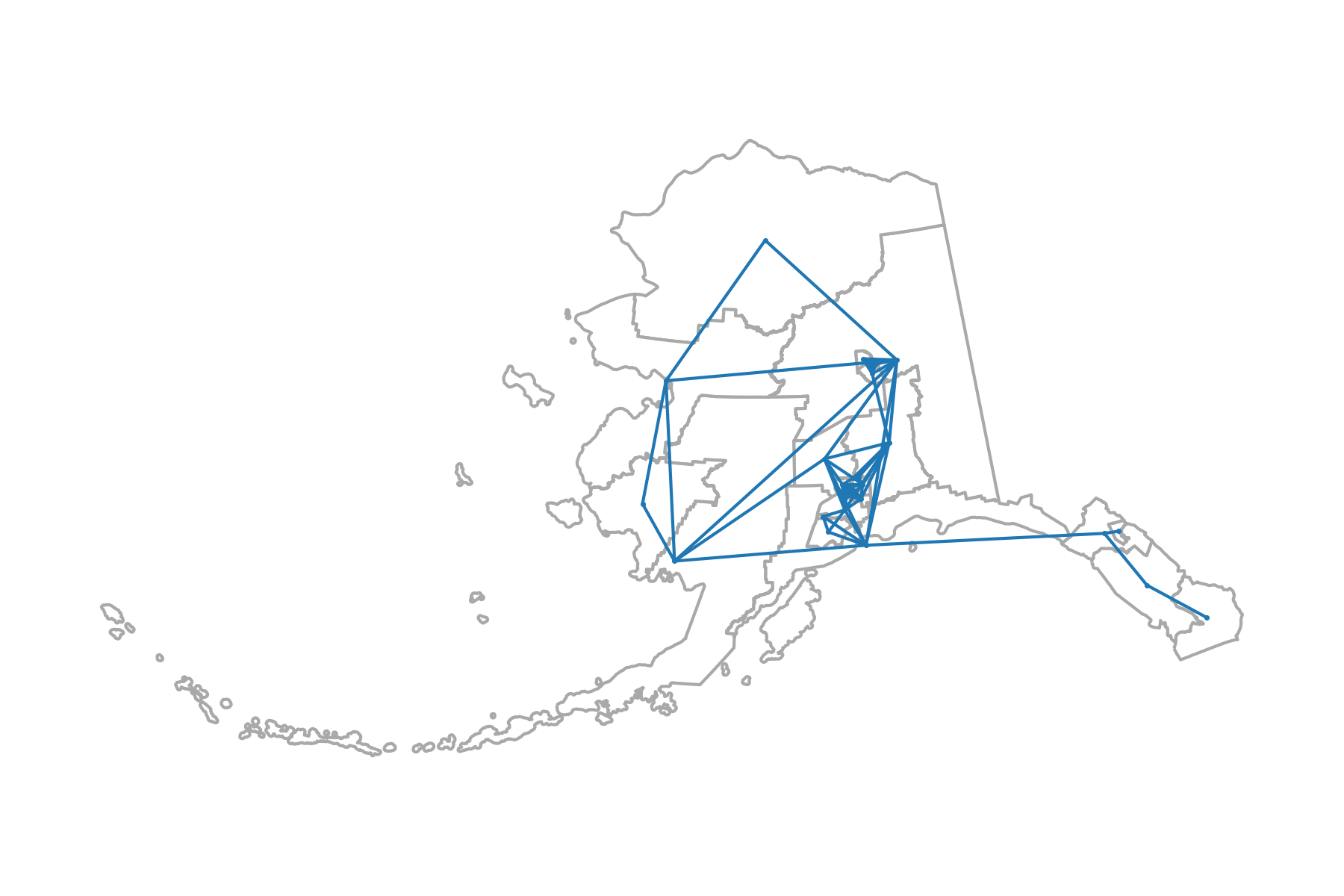}};
\node at (-11,-2.3) {\small Alaska House};
\node at (-11,-3) {\small Districts};
\node at (-5.5,0) {\includegraphics[height=1.1in]{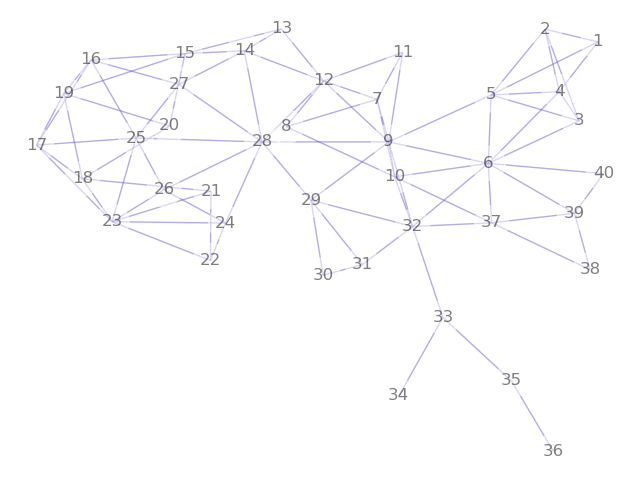}};
\node at (-5.5,-2.3) {\small Tight};
\node at (-5.5,-3) {\small (89 edges)};
\node at (0,-2.3) {\small Restricted};
\node at (0,-3) {\small (92 edges)};
\node at (5.5,-2.3) {\small Permissive};
\node at (5.5,-3) {\small (100 edges)};
\node at (0,0) {\includegraphics[height=1.1in]{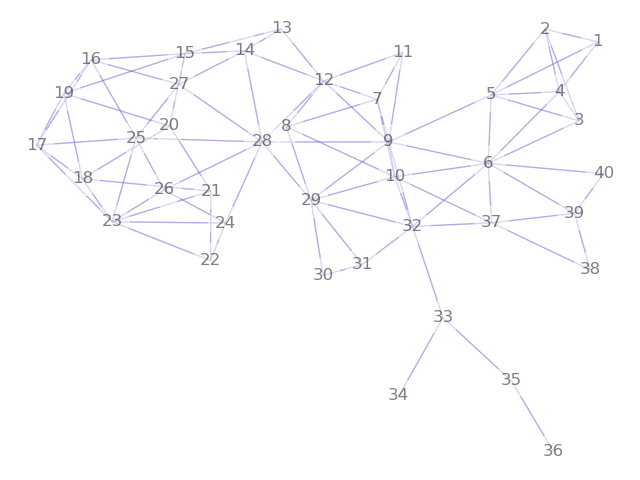}};
\node at (5.5,0) {\includegraphics[height=1.1in]{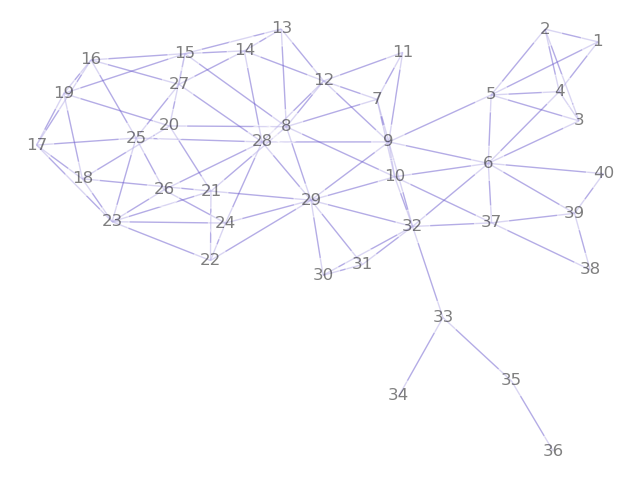}};
\end{tikzpicture}    
\caption{The three Alaska dual graphs.}
    \label{fig:AKDG}
\end{figure}





%
\subsection{Markov chains for generating 
 alternative House plans}

Markov chain Monte Carlo, or MCMC, is the leading method
in scientific computing for searching large spaces and studying properties of complex systems.  
Numerous research groups
now use MCMC implementations to study the universe of possible
districting plans, once the basic units have been set
\cite{mattingly,pegden,va-report}.
We use the open-source software package {\sf GerryChain},
created by the Voting Rights Data Institute \cite{gerrychain}.

This algorithm generates new plans iteratively, making modifications
to the districting assignments of some of the units at each step. To choose
how to modify districts at random, we use a proposal called Recombination ({\sf ReCom}),
 introduced in \cite{va-report}.\footnote{The scientific advantage of using Markov chains to sample
 districting plans is that they have a theoretical guarantee of producing representative samples
 (with respect to their stationary disributions) if run for long enough. In our case, we run the chains 
 until we obtain strong heuristic evidence of mixing, which is a common and effective standard in scientific
 computing.  See \cite{va-report}.}
At each step, this proposal merges the units of two adjacent districts and
then uses a balanced cut of a spanning tree to generate a new division of the merged districts.  
The user
can elect to impose hard constraints in the form of 
requirements for plans, or can choose a weighted random walk
that preferentially selects plans with properties deemed
to better comport with the districting principles.  
The tree method itself promotes the selection of compact districts, so the plans generated in this way 
tend to have comparable compactness statistics to human-approved plans and to comfortably pass the 
``eyeball test'' for district shape.

We ran our Markov chains 
on Alaska's
441 precincts as basic units, seeking new legally valid
plans for forming them into 40 House districts. 
As outlined above in \S\ref{AK-rules}, the law requires
that the districters aim to produce equipopulous, compact,
and contiguous districts.

The ideal population of a House district
is 710,231/40, or between 17,755 and 17,756 people. 
By federal law 
and common practice, legislative districts can deviate by up to 5\%
from ideal size without a special reason, so we have imposed
that limit on population balance, 
allowing 16,868--18,643 people per district.  The same level of population
balance was imposed on the Senate ensembles.

We generated  ensembles of 100,000 distinct House and Senate districting plans, varying the 
definitions of contiguity (tight/restricted/permissive).  
A district-level dual graph of the sampled plan was pulled every step.
Using FKT, we counted the number of matchings for each districting plan and edges between House districts, 
and we stored plans with extreme matching statistics.
The goal was to learn whether the new plans could have 
markedly different partisan outcomes, either on their own or when matched to 
form a Senate plan, from the current districts.

Replication code for producing and analyzing these ensembles is available on GitHub \cite{Alaska-Repo}.


\section{Alternative matchings in current House plan}\label{sec:altmatch}

In this section, we evaluate the currently enacted House plan by generating all of the potential matchings and computing their partisan behavior under our chosen election data. We start by computing the number of matchings for each of our notions of adjacency, reported in Table~\ref{tab:AK_timing}, using the FKT algorithm described in Appendix B. 
This already  highlights the fact that toggling  a small number of edges in the Alaska House dual graph (due only to reasonable interpretations
of water adjacency) can change the number of perfect matchings very substantially; in this case, the number of matchings jumps by nearly a factor of eight.\footnote{It is worth emphasizing that 
the number of matchings is sensitively dependent on the precise edge structure as well as simply the number of edges.  This fact is explored
below in \S\ref{sec:rematching}.}

\begin{table}[ht]
    \centering
    \begin{tabular}{|r|c|c|c|}
    \hline
  Alaska  & Tight&Restricted&Permissive  \\
  \hline
  \hline
  Dual edges & 89& 92&100\\
  \hline
  Matchings& 14,446&29,289&108,765\\
  \hline
  FKT runtime     &0.022  sec&0.022 sec&0.027 sec\\
  \hline
  Prune--and--choose runtime  & 14.2 sec& 28.5  sec&105.1 sec\\
  \hline
    \end{tabular}
\caption{On Alaska, FKT (which only counts matchings) runs in a fraction of a second; the prune-and-choose algorithm (which stores matchings) runs in under two minutes.\label{tab:AK_timing}}
\end{table}

We then use the prune-and-choose algorithm described in Appendix C to generate each of the possible pairings. 
For each matching, we evaluate its performance under each of the four  election 
treatments,
comparing the number of Senate districts
with a Democratic majority in the actual plan to the average number over the Senate plans formed by all possible matchings.
It bears emphasizing that the  
Democratic seats reported across the table refers to the number of 
Senate districts in which 
Galvin votes outnumbered Young votes 
for Congress, or Begich votes outnumbered
Dunleavy votes for Governor.

\begin{table}[!ht]
    \centering
    \begin{tabular}{|c|c|c|c|c|}
    \hline
&\multicolumn{4}{c|}{D Senate districts}\\
    &  Enacted & Tight&  Restricted &  Permissive  \\
     \hline\hline
    Cong18-N&6&6.90&6.78&6.52 \\
    \hline
    Cong18-A&7&8.30&8.18&7.93 \\
    \hline
    Gov18-N  &7&7.69   & 7.57&7.33 \\
    \hline
    Gov18-A&7&8.24&8.13& 7.88\\
    \hline
    \end{tabular}
    
    \vspace{.1in}
    
      \begin{tabular}{|c|c|c|c|c|}
    \hline
    &\multicolumn{4}{c|}{Competitive districts}\\

    &  Enacted & Tight& Restricted &  Permissive \\
        \hline\hline
     Cong18-N   &13&12.12&11.90 &11.91  \\
     \hline
        Cong18-A&12&12.33&12.26 & 12.54\\
        \hline
        Gov18-N& 11&10.54&10.55&10.21 \\
        \hline
        Gov18-A& 10&9.78 & 9.68 & 9.68\\
    \hline
    \end{tabular}
    
    \caption{Partisan outcome and competitiveness
    in the current (Enacted) Senate plan compared to the average over alternative matchings. 
    A competitive district is     defined here as one with
    a D share  between 40 and     60 percent.  Here, every partisan outcome is more favorable
    to Republicans than the neutral expectation.  Compare Table~\ref{tab:ensemble}, which
    varies the underlying House plan and shows the opposite partisan tendency.\label{tab:matchings}}
\end{table}

Comparing the enacted plan against all possible pairings  does indeed find
a small Republican tilt, falling
approximately one seat to the Republican side
of the typical matching but certainly does not appear to be a significant outlier.\footnote{The actual Senate composition also has six or seven Democrats, depending on how you count:
seven state Senators were elected as Democrats, but Sen. Lyman Hoffman caucuses with the Republican majority.}
(At the same time, we can observe the substantial effect of discarding absentee ballots:  it  shifts outcomes
towards Republicans by about 1 seat.)
The  histograms in Figure~\ref{fig:matchhist} add detail by showing the full distribution of 
Democratic seats with respect to each race.
\begin{figure}[!ht]
 \centering
 \begin{tikzpicture}[scale=.7] 
	 \node at (0,0)   {\includegraphics[width=2in]{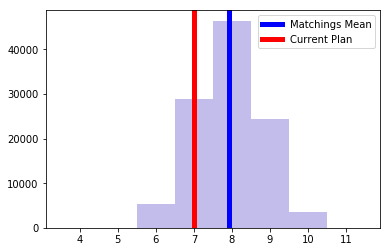}};
 	\node at (0,-3) {\small Cong18-A D Senate seats};
 	\node at (7.5,0)   {\includegraphics[width=2in]{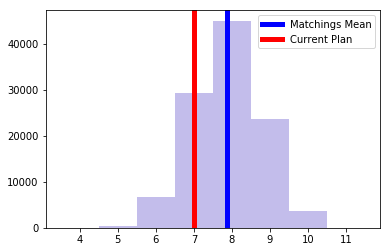}};
 	\node at (7.5,-3) {\small Gov18-A D Senate seats};
 	\end{tikzpicture}
 \caption{The number of Senate districts with a D majority, as the matching varies across the permissive set. The green line marks the average number of Democratic seats over all matchings and the red line shows the outcome in the current plan, showing a one-seat
 advantage for Republicans in the current matching and a four- to 
 five-seat swing overall.
 These two histograms look nearly identical for the different 
 elections, despite the substantial differences
 in how the vote was distributed.}
 \label{fig:matchhist}
\end{figure}

For even more granular detail, at the level of individual districts,
we can study  box-and-whisker plots (Figure~\ref{fig:boxplots}). 
In these images, the districts are ordered from 
lowest to highest Democratic vote share in order to make them 
comparable over the ensembles.
The boxes show the 25th-75th
percentile range and the whiskers
show every result achieved over
the full set of (permissive) matchings.
 Similar histograms and boxplots for the remainder of the elections and dual graphs are available in our supplemental material \cite{Alaska-Repo}.

\begin{figure}[!ht]
 \centering
 \begin{tikzpicture} 
	 \node at (0,6.3)   {\includegraphics[width=4.2in]{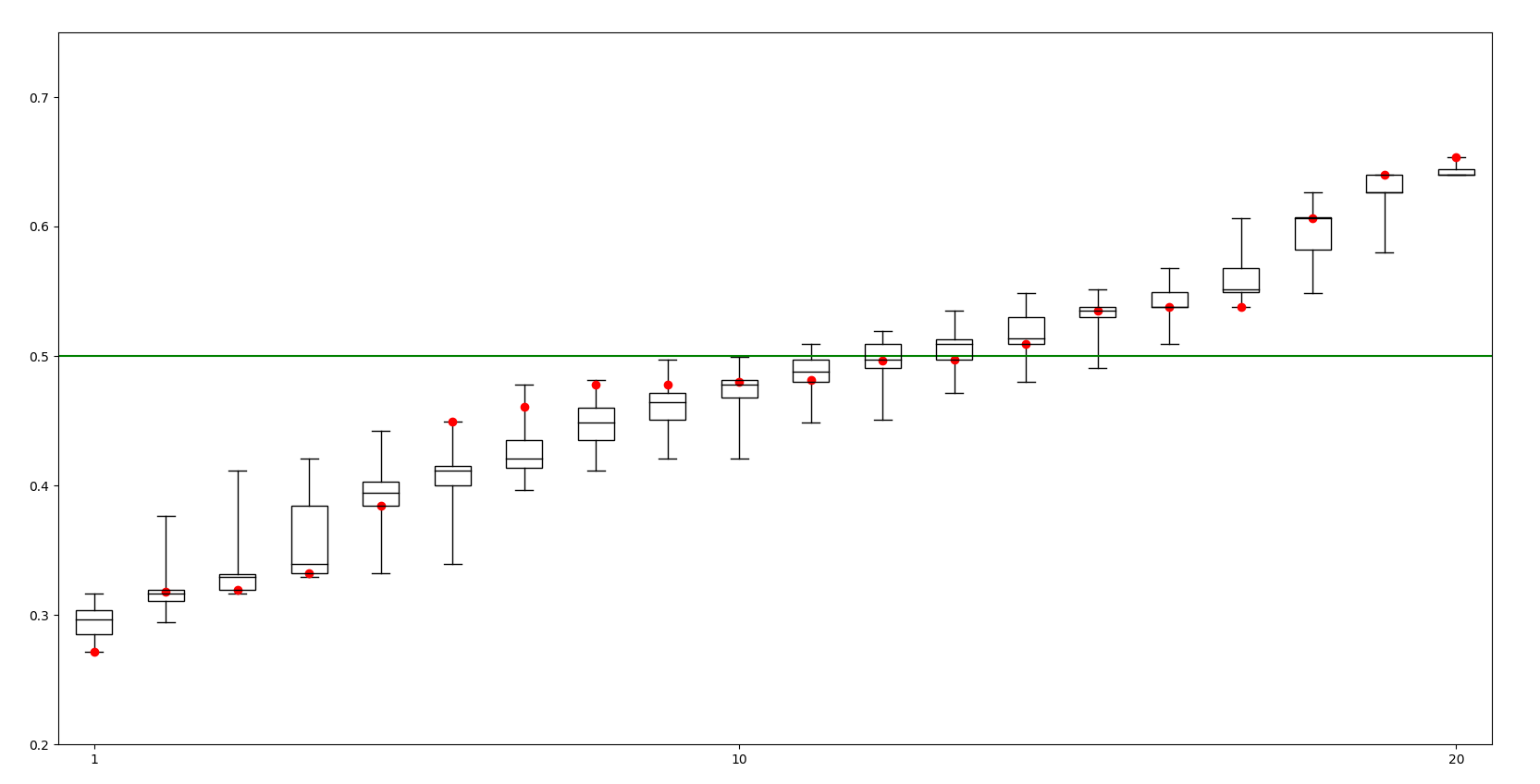}};
 	\node at (0,3.3) {\small Cong18-A Dem share by Senate district};	
	\draw [blue] (1.05,6.1) rectangle (1.53,7.1);
 	\node at (0,0)   {\includegraphics[width=4.2in]{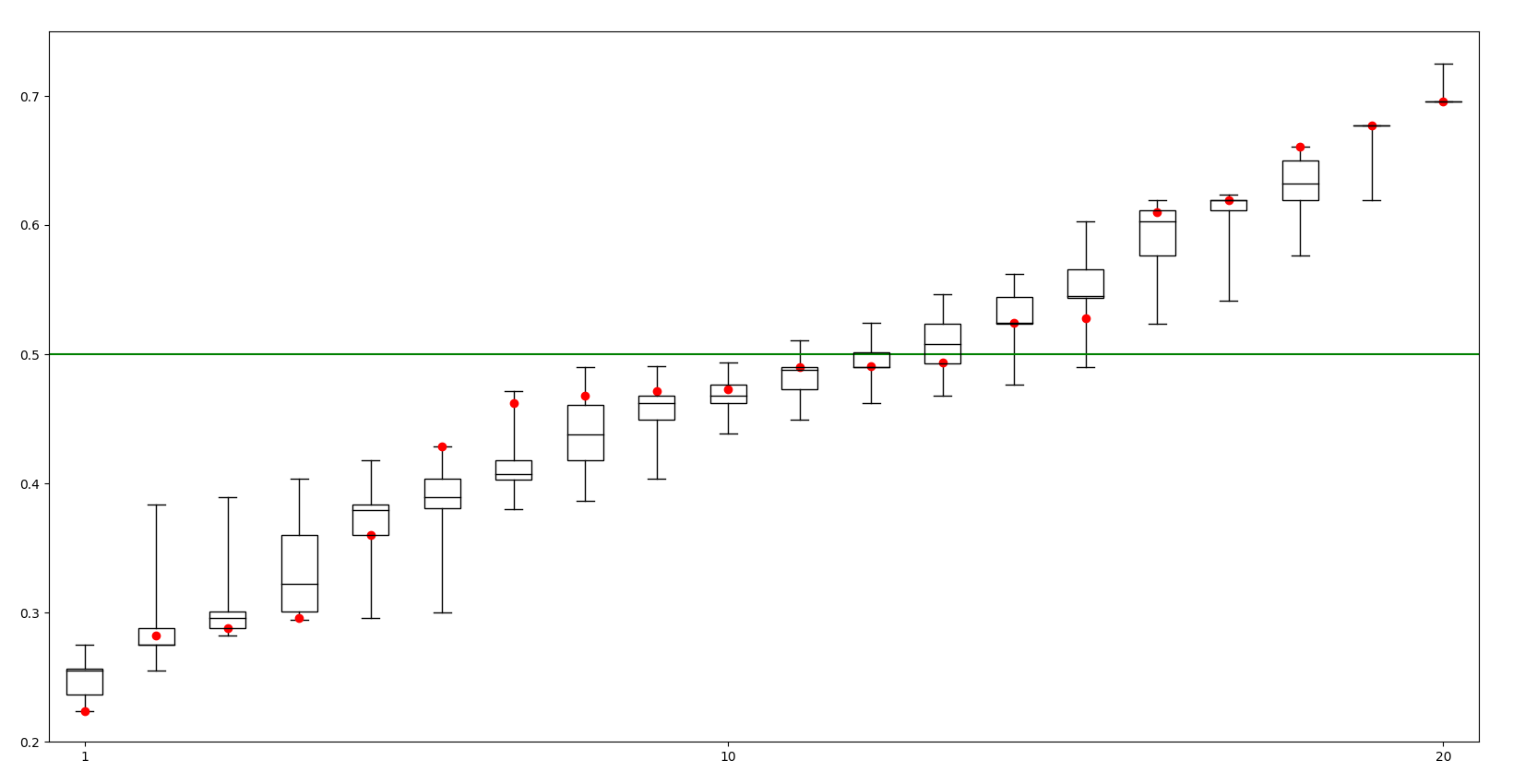}};
 	\node at (0,-3) {\small Gov18-A Dem share by Senate district};
		\draw [blue] (1.053,-.2) rectangle (1.555,.855);
 	\end{tikzpicture}
 \caption{Democratic vote share in current Senate districts (red dots), compared to range in comparable districts over the full set of matchings (box and whiskers). With district-by-district detail, the differences between the two elections' voting patterns are more visible.
For instance, the 13th-indexed districts in the state 
have a Galvin (Congressional) share and a Begich (Governor) share just under the median of the respective ensembles,
 while nearly 75\% of the ensemble in each case had a Democratic majority in the corresponding district.
Where boxes have degenerated to a single value, it is because some matchings  are forced, thinning the number of  possibilities.  
 \label{fig:boxplots}}
 \end{figure}
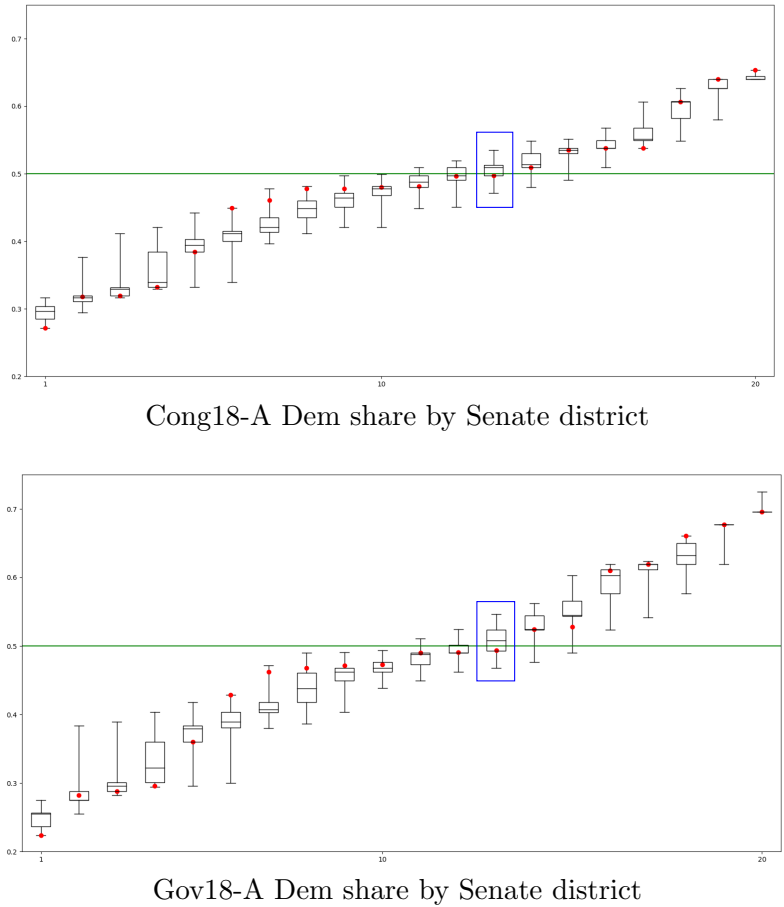

\newpage
\section{Alternative House and Senate districts}\label{sec:chain}

\subsection{Partisan outcomes}
Using the Markov chain ensembles of 100,000 plans each as a neutral counterfactual for drawing districts,
we first report the number of House districts out of 40 with more Democratic than Republican votes.

\begin{table}[!h]
    \centering
    \begin{tabular}{|c|c|c|c|c|}
    \hline
&\multicolumn{4}{c|}{D House districts}\\
    &  Enacted & Tight& Restricted &  Permissive \\
     \hline\hline
    Cong18-N&15&14.10&13.93&14.28\\
    \hline
    Cong18-A&18&16.35&16.33&16.73\\
    \hline
    Gov18-N&15&14.73&14.55&14.44\\
    \hline
    Gov18-A&20&17.86&17.63&17.65\\
    \hline
    \end{tabular}

\vspace{.1in}

    \begin{tabular}{|c|c|c|c|c|}
    \hline  
    &\multicolumn{4}{c|}{Competitive districts}\\
        &  Enacted &  Tight&Restricted &  Permissive  \\
              \hline      \hline
        Cong18-N  &  24&25.51&25.01&25.23\\
            \hline
        Cong18-A  &24&24.85&24.62&24.88\\
            \hline
    Gov18-N  &19&21.03&20.60&20.59\\
            \hline
   Gov18-A   &19&20.83&20.57&20.52\\
      \hline
\end{tabular}      
    \caption{The number of House districts
 with a D majority and the number
 of competitive districts, as the House plan  varies.
     A competitive district is 
    defined as one in which 
    the D share is between 40 and 
    60 percent.
Compare Table~\ref{tab:matchings}.}
    \label{tab:ensemble}
\end{table}

Beyond the averages, we can view
the full histograms to gauge the 
extent to which the current plan 
is an outlier.

 \begin{figure}[!h]
 \centering
 \begin{tikzpicture}[scale=.7]
 	\node at (0,0)   {\includegraphics[width=2in]{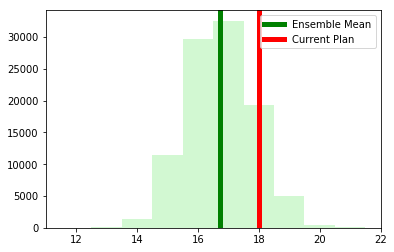}};
 	\node at (0,-3) {\small Cong18-A D House seats};
 	\node at (7.5,0)   { \includegraphics[width=2in]{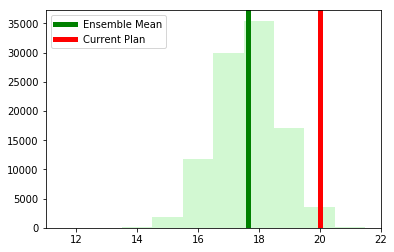}};
 	\node at (7.5,-3) {\small Gov18-A D House seats};
 	\end{tikzpicture}
 \caption{The number of  districts with a D majority in the indicated election, over the ensemble of (permissive) House 
 plans.}
 \end{figure}
 
 Recall that the current majority House caucus includes 24 members: 16 members who were elected as Democrats, seven members who were elected as Republicans, and one who was elected as an Independent.

 \begin{figure}[!h]
 \centering
 \begin{tikzpicture}[scale=.7]
 	\node at (0,0)   {\includegraphics[width=2in]{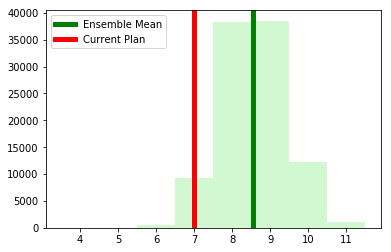}};
 	\node at (0,-3) {\small Cong18-A D Senate seats};
 	\node at (7.5,0)   { \includegraphics[width=2in]{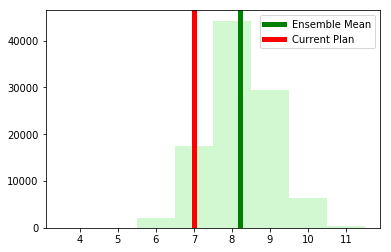}};
 	\node at (7.5,-3) {\small Gov18-A D Senate seats};
 	\end{tikzpicture}
 \caption{The number of  districts with a D majority in the indicated election, over the ensemble of (permissive) Senate 
 plans.}
 \end{figure}

The Senate ensemble gives another interesting vantage point.  With respect to both vote patterns,
the bulk of plans assembled by the Markov chain process have 7-10
Democratic Senate seats, and the full range observed in the ensemble is 6-11.  
Compare this to simply matching the current House plan, where we can fully exhaust the possibilities
instead of sampling.  Matchings give us Democratic seat outcomes of 6-10 in either vote pattern,
with a small number that achieve a 5-seat outcome against the Governor returns.  
This means that mere control over the matchings gives essentially just as much partisan latitude
as the right to draw plans from scratch with the most permissive notion of precinct adjacency.  

\subsection{Native population}\label{sec:native}

We also find that the 
number of districts with an Alaska Native population
majority is typically 3-4 in our randomly produced House 
plans, compared to three in the current House plan.
Furthermore, the ability to form districts more 
permissively across water makes a very noticeable difference, boosting 
the likelihood of forming a fourth
majority-Native district by random 
selection.

\begin{table}[!h]
\centering
\begin{tabular}{rccc}
 Number of majority-Native House districts &2&3&4  \\
\hline
    Permissive ensemble & 444&53,596&45,960\\
    Restricted ensemble & 1,053&97,069&1,878\\
    Tight ensemble&1,135& 97,507&1,358
\end{tabular}
\end{table}


\subsection{Re-matching the new plans}
\label{sec:rematching}

Over each ensemble of  100,000 House plans, we 
computed the number of dual edges 
and the number of perfect matchings.  
The (nonzero) numbers of matchings varied from 74--165,344 (tight),
42--194,588 (restricted), and 852--961,176 (permissive).  
To a first approximation, more edges means more
matchings, but the scatterplots in Figure \ref{fig:scatter}
show that there is also substantial dependence on the specific placement of the edges.\footnote{One notable feature of Fig
\ref{fig:scatter} is the prevalence of plans with low numbers of matchings.  It is easy to construct graphs  with any number of edges and zero matchings, simply by 
having any two leaves (vertices of degree one) connected to a single common neighbor---and this can easily occur by chance.
There were no 
matchings at all in 2319, 4274, and 3504 of the dual graphs found by the ensembles, respectively.
It is similarly easy to randomly construct graphs 
with very few matchings simply by having many leaves and thus
many forced edges.  On the other hand,
there are graphs with $2n$ vertices,
roughly $n^2/2$ edges, and only a single matching:  
start with a complete graph on $n$ vertices (i.e., with all edges
present) and add a single leaf connected to each of those.
Each leaf vertex is forced to match to its unique neighbor,
leaving no more vertices to pair.}

\begin{figure}[ht]
\centering
\begin{tikzpicture}
\node at (0,4.2) {\includegraphics[width=1.9in]{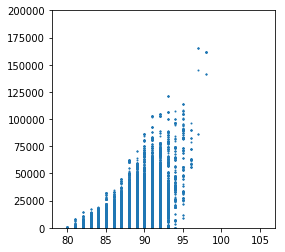}}; 
\node at (0,0) {\includegraphics[width=1.9in]{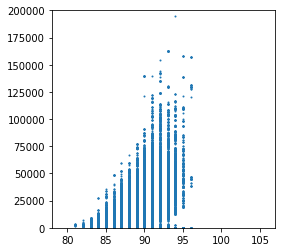}};
\node at (6,2.25) {\includegraphics[width=1.9in]{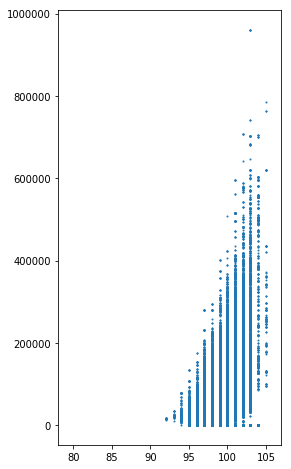}};
\end{tikzpicture}
\caption{The relationship between the number of edges
($x$ axis) and the number of matchings ($y$ axis).
Each point is a House plan, varying
over the tight (top left), restricted (lower left), and permissive (right) ensembles.\label{fig:scatter}}
\end{figure}

To illustrate the sensitive dependence of the number of matchings on the precise topology of the graph, we focus on three examples found in the permissive House ensembles in Figure \ref{fig:extremes}. The leftmost dual graph has 103 edges and 961,176 matchings while the next one has more edges but less than half the number of matchings. The disparity in these matching counts is almost entirely due to the fact that there are two ways to pair the districts in the ``panhandle'' of the 103-edge 
graph---[(29,30), (39,40)] or [(29,39), (30,40)]---compared to the unique pairing 
 [(17,40), (35,36)] in the corresponding region of the 105-edge graph. 
This accounts for a doubling in the number 
of overall matchings, assuming a comparable
number of ways to match the remaining 36 vertices.
Forced pairings  play an important role in plans with few matchings. The rightmost graph shows an example with 101 edges but only 852 matchings, the lowest nonzero
number ever observed. This is due to the many forced pairs---[(11,13), (12,34), (2,3), (1,36), (4,33), (8,24), (7,26), (17,25)]---which limits the flexibility in pairing the remaining vertices. 

\begin{figure}[!h]
    \centering
\begin{tikzpicture}[scale=.7]
\node at (-6,0) {\includegraphics[width=1.5in]{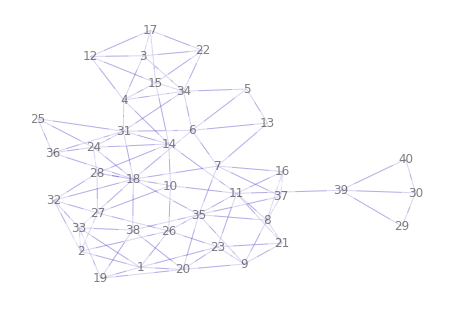}};
\node at (-6,-2) {\small 103 edges};
\node at (-6,-2.5) {\small 961,176 matchings (most)};
\node at (0,0) {\includegraphics[width=1.5in]{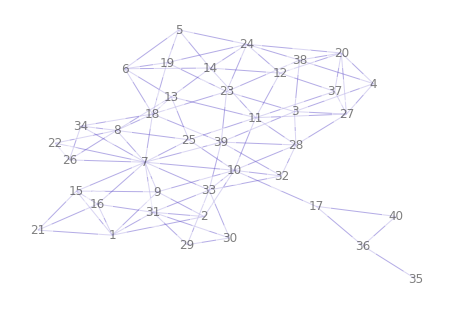}};
\node at (0,-2) {\small 105 edges (most)};
\node at (0,-2.5) {\small 436,220 matchings};
\node at (6,0) {\includegraphics[width=1.5in]{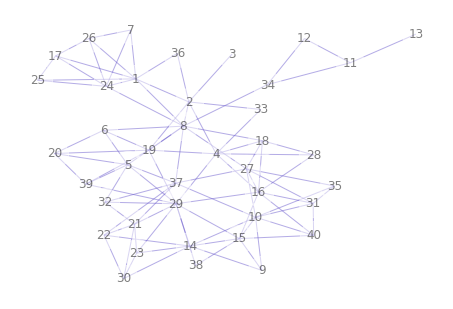}};
\node at (6,-2) {\small 101 edges};
\node at (6,-2.5) {\small 852 matchings (fewest)};
\node at (-3.5,-6) {\includegraphics[width=2.5in]{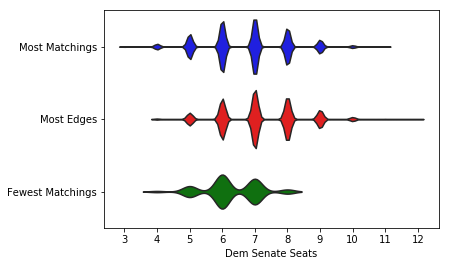}};
\filldraw [blue!70!black] (4,-5.8) circle (9.88*.15); 
\filldraw [red!85!black]  (4,-5.8) circle (6.6*.15);
\filldraw [green!50!black] (4,-5.8) circle (.29*.15) ;
\node at (4, -7.9) {\scriptsize Violin plot scale};
\end{tikzpicture}    
\caption{A selection of three House plans from our ensembles whose dual graphs have
various extremal properties.  The violin plot shows the number of Democratic districts with respect to Gov18-A vote
data, and the colored regions show the relative sizes of the matching sets.  \label{fig:extremes}}
\end{figure}
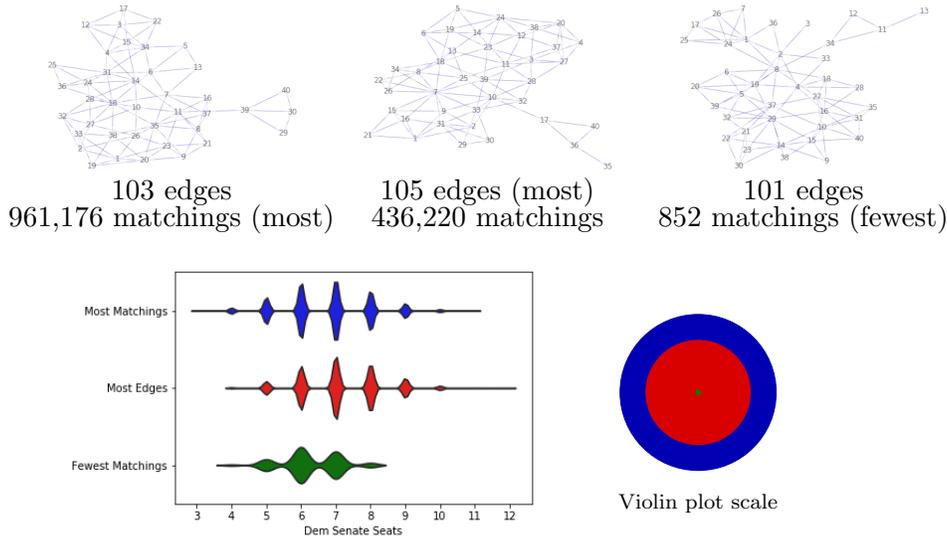

The analysis demonstrates that the matchability of the underlying House plan can have a significant downstream partisan impact on 
the Senate plans that can be formed. Figure \ref{fig:extremes} shows examples of this behavior by comparing the distributions over the possible perfect matchings for three plans from the permissive ensemble.  
For each of the three House plans, a typical perfect matching has 5-8 Democratic districts out of 20.  However, by choosing 
a House plan with more district adjacencies or more matchings, it is possible to get as few as 3 or as many as 12. Investigating the full flexibility allowed to a mapmaker who controls the House district drawing process is an interesting question for further research.

\section{Conclusions}

Numerous studies have sought to quantify the partisan
advantage secured by the selection of a particular
districting plan.
In that vein, we find that the current Alaska House plan favors Democrats by an estimated 1-2 seats out of 40 when compared to other (contiguous, compact, population-balanced) ways of forming  districts from precincts.  
The core of the paper, however, is a novel application
of rigorous mathematics to redistricting in the case
of a nesting rule for state legislatures.  
For that, we can apply the theory of perfect matchings
of graphs, learning that the Alaska Senate plan secures a Republican advantage of 1 seat out of 20 when compared to other ways to match the House districts.  Other findings:
\begin{itemize}
\item The choice of matching of a fixed House plan gives as much latitude to control partisan outcomes as 
drawing a Senate plan from scratch: approximately a five-seat swing out of 20.
\item The significant number of absentee/early/provisional ballots in Alaska 
skew markedly Democratic.  Different choices of how to assign
them to precincts will impact findings about the consequences of 
moving district boundaries, and should be further studied. However, this has no effect on our analysis of matchings.
\item Well-chosen statewide races, in this case the 2018 Governor and Congressional elections, gave partisan
measurements that are closely compatible with each other and qualitatively concordant with the Legislative outcomes.
\item Contiguity rules are not completely straightforward, and can have a major role in shaping the space of 
districting possibilities.  For instance, permissive water adjacency makes nearly half of neutrally generated House plans have
a fourth majority-Native district, while less than 2\% of plans do with more restricted adjacency (\S\ref{sec:native}).
\end{itemize}


\newpage
\appendix

\begin{center}
    {\LARGE\bf Supplement to Mathematics of Nested Districts: The Case of Alaska}
\end{center}

This supplement contains additional plots and technical descriptions of the algorithms used in the paper for enumerating and sampling perfect matchings. Appendix A shows reference figures of the state House dual graphs for the non-Alaska states with strict nesting rules. Further information about the FKT algorithm is given in Appendix B while Appendix C introduces the new Prune--and--Choose algorithm and provides a proof of correctness. In Appendix D, these algorithms are applied to all states with a nesting rule to compare the relative scales of the enumeration problem. Finally, in Appendix E we implement and validate a uniform sampling procedure for matchings that can be applied even in the states where generating all of the matchings would be computationally infeasible.

\newpage

\section{Dual Graphs of States with Nesting Rules}
In this appendix we show the dual graphs for the other  states that require two-to-one nesting. Corresponding plots for Alaska are shown in Figure 7 of the main text. The left-hand column
    shows the House districts, with the dual graph overlaid, and the right-hand column shows a nearly-planar embedding with accurate district labels.
Given these House districts, valid Senate plans in these states correspond to perfect matchings of these graphs.

\begin{figure}[!h]
    \centering
    {\bf Illinois}\\
    \includegraphics[height=2in]{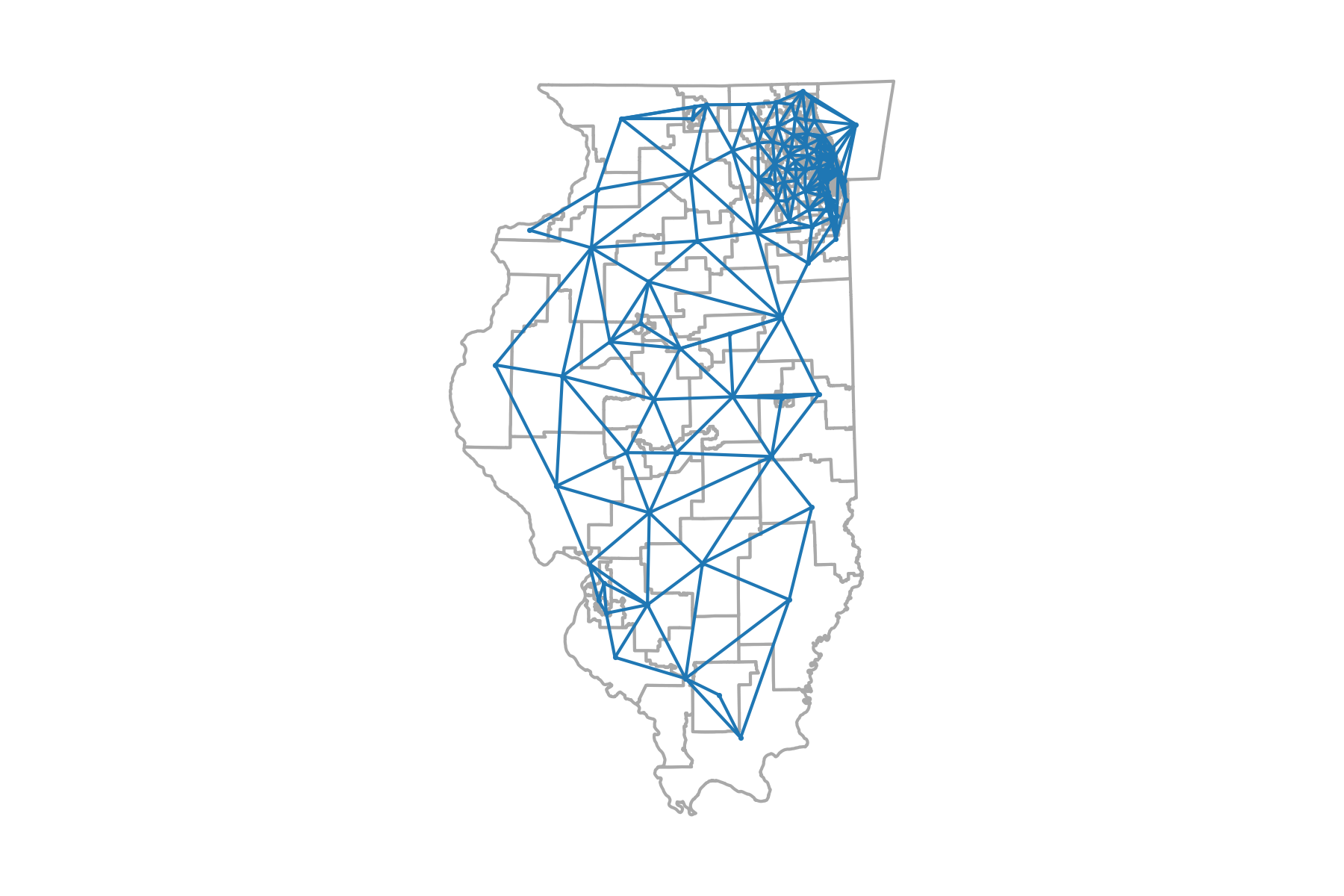}\quad \includegraphics[height=2in]{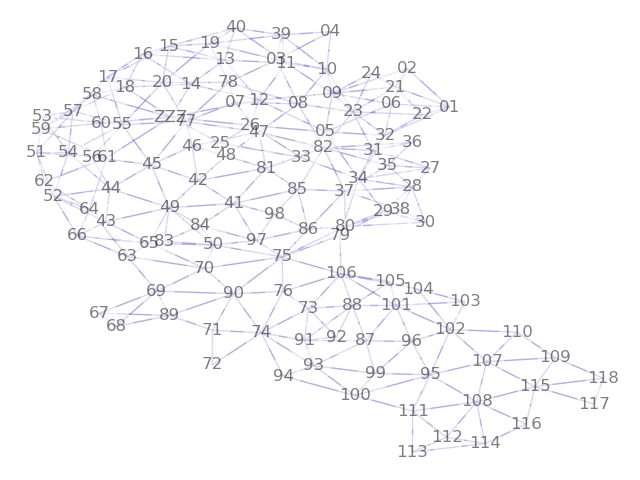}\\
    {\bf Iowa}\\
    \includegraphics[height=2in]{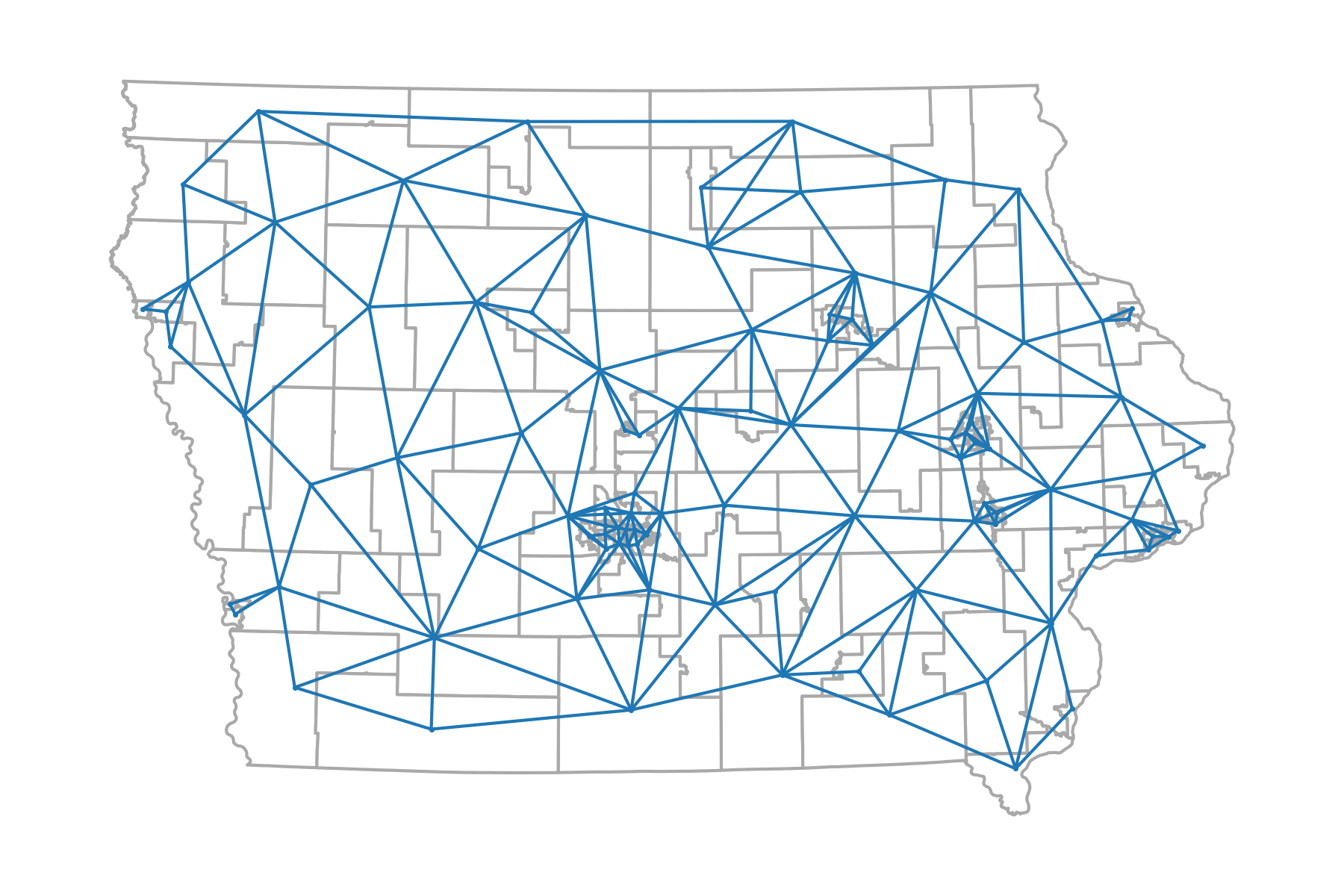}\quad \includegraphics[height=2in]{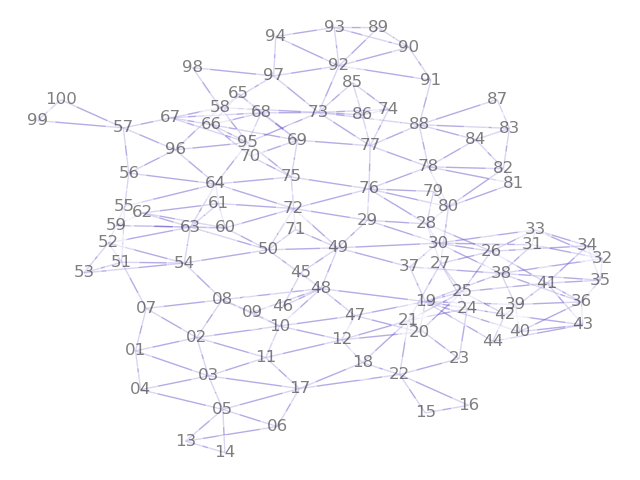}\\
    {\bf Minnesota}\\
    \includegraphics[height=2in]{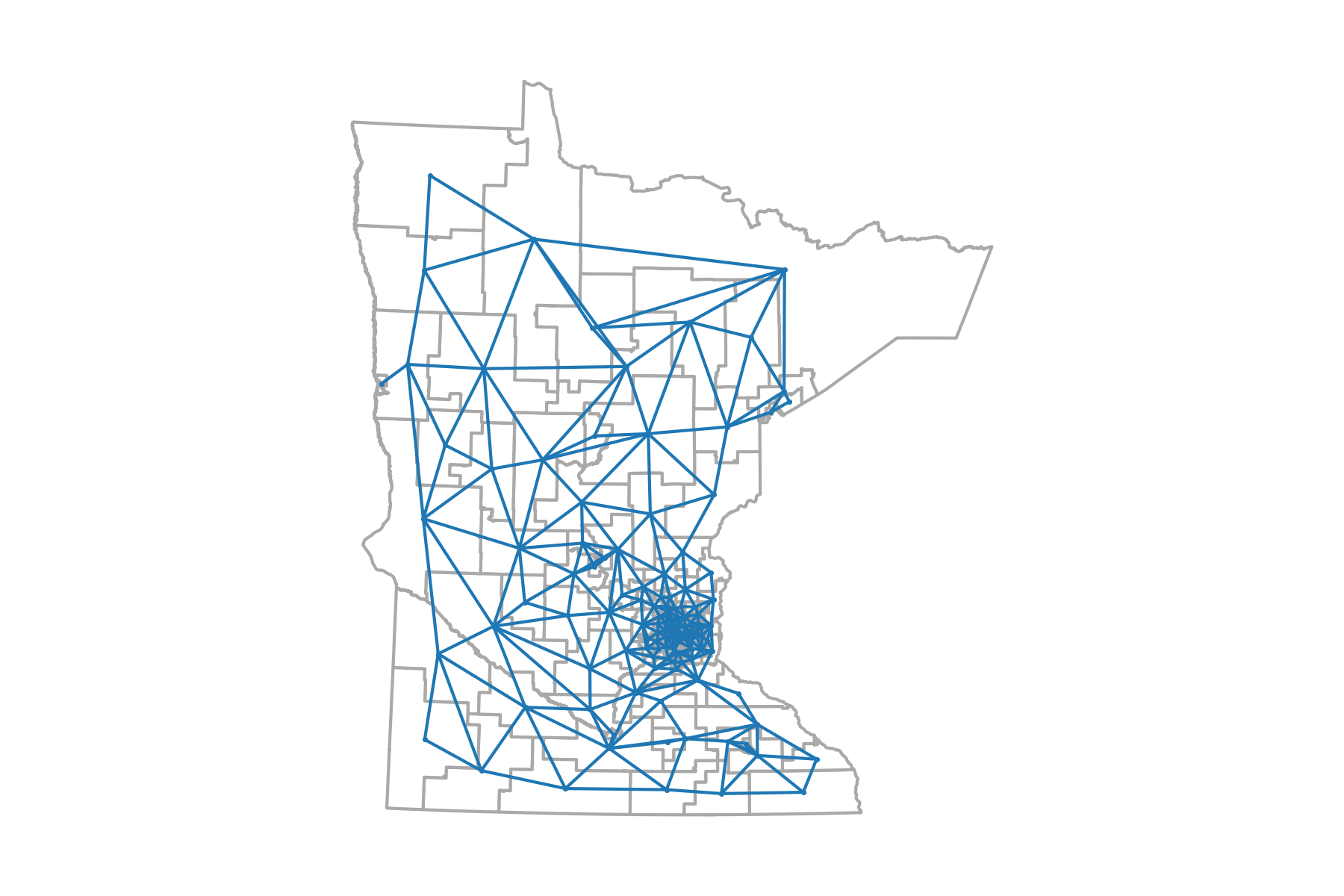}\quad \includegraphics[height=2in]{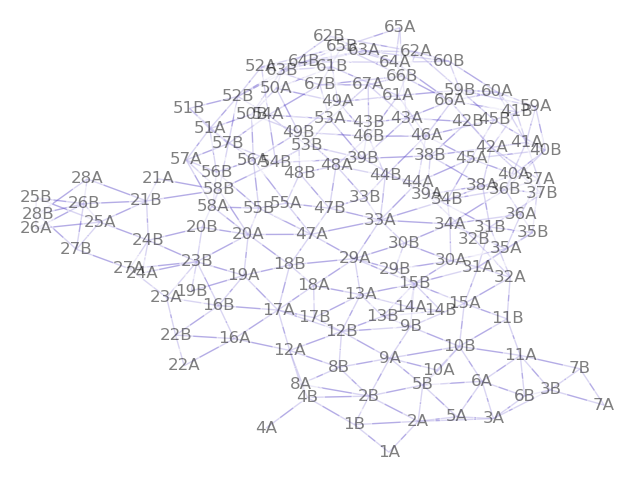}\\

\end{figure}

\newpage

\begin{figure}[H]
    \centering
    {\bf Montana}\\
    \includegraphics[height=2in]{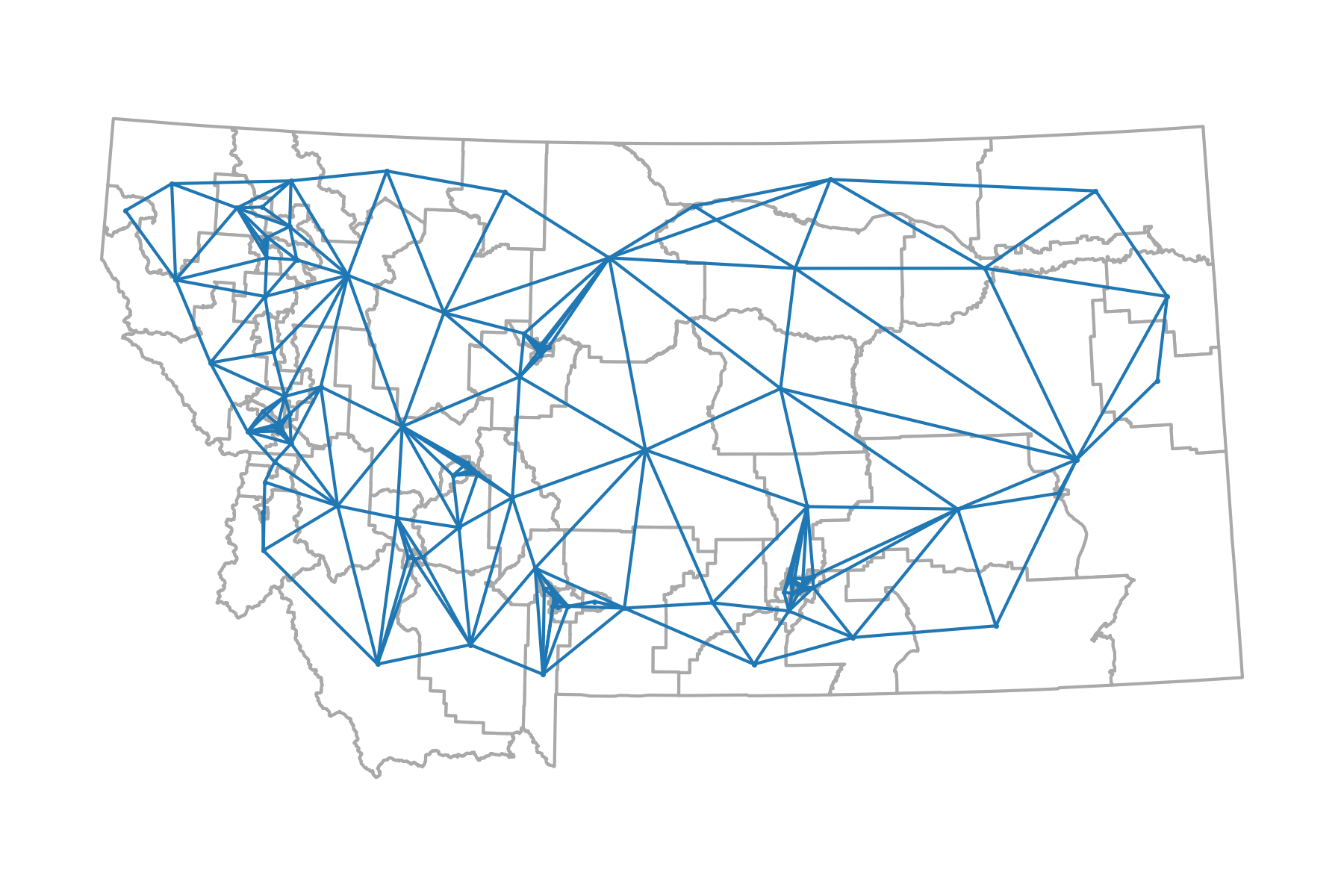}\quad \includegraphics[height=2in]{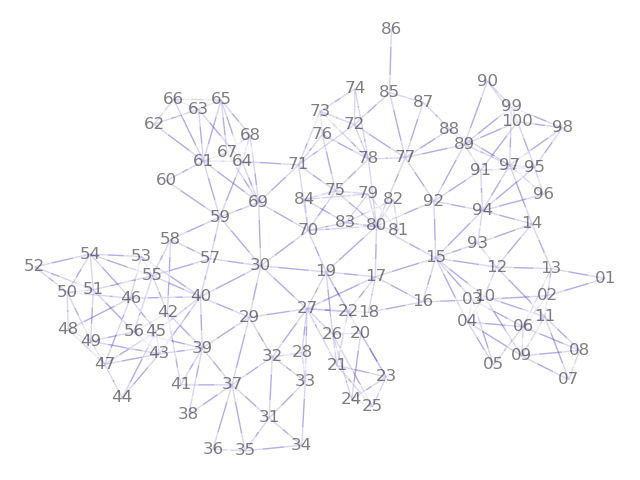}\\
    {\bf Nevada}\\
    \includegraphics[height=2in]{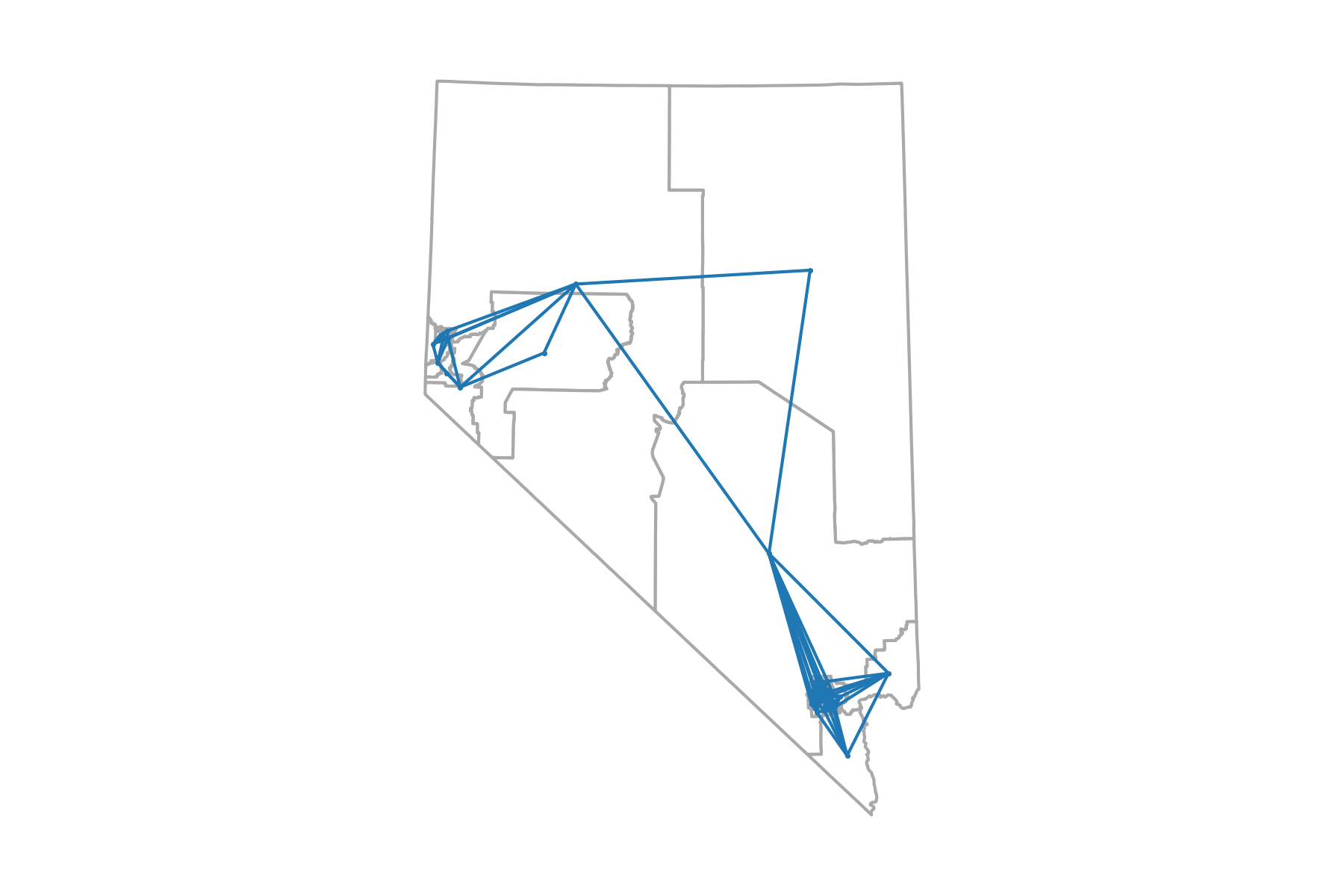}\quad \includegraphics[height=2in]{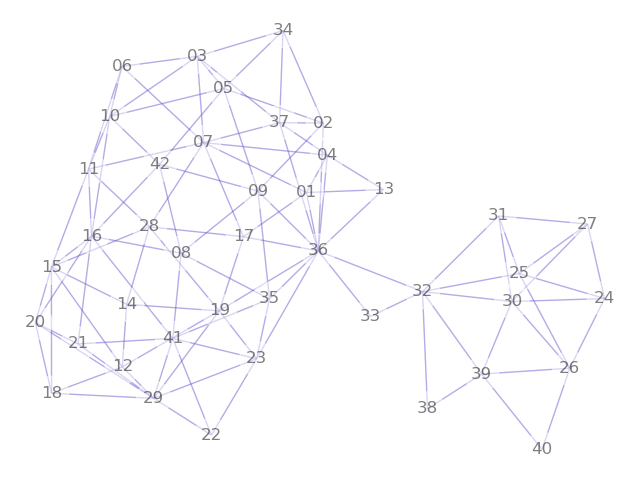}\\
{\bf Oregon}\\
    \includegraphics[height=2in]{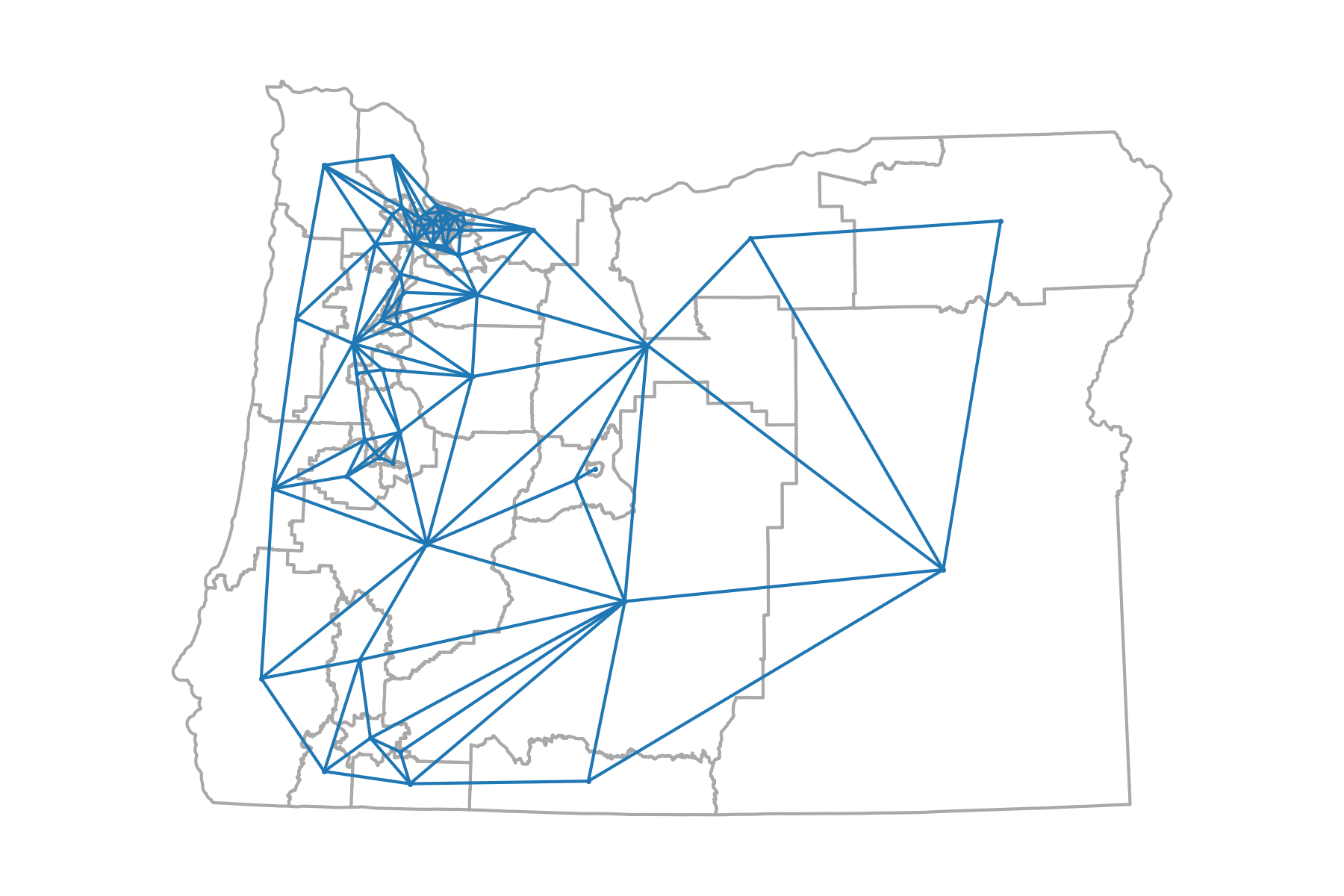}\quad \includegraphics[height=2in]{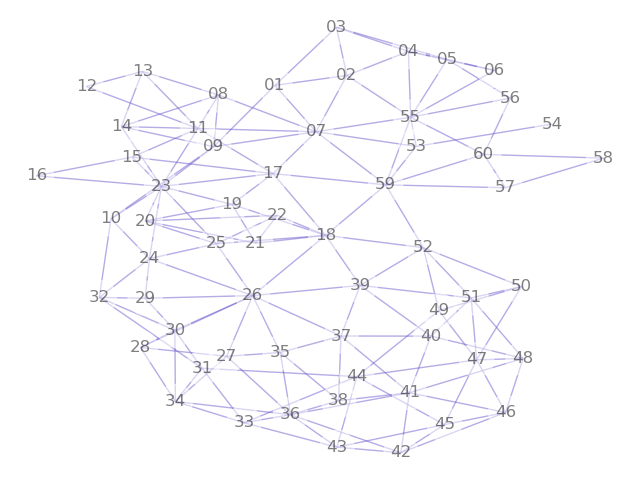}\\
    {\bf Wyoming}\\
    \includegraphics[height=2in]{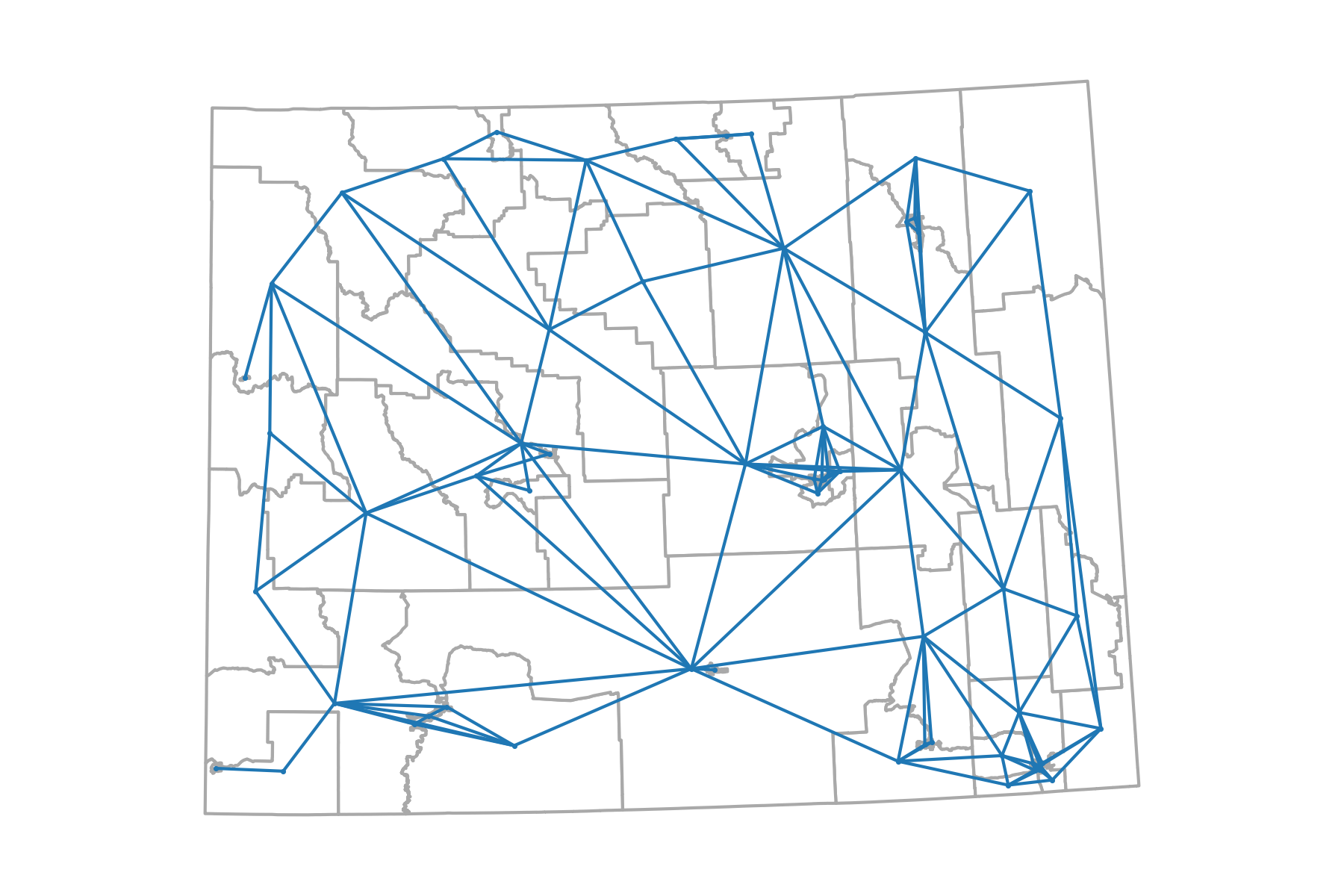}\quad \includegraphics[height=2in]{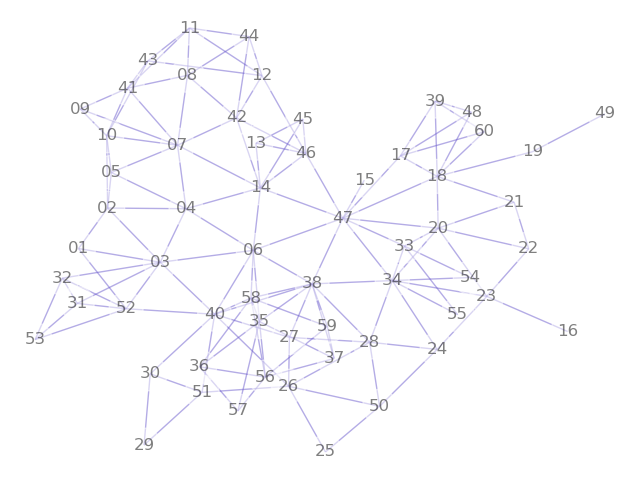}\\

\end{figure}



\newpage

\newpage

\section{FKT algorithm for enumerating perfect matchings}
\label{sec:FKT}

As described in Section 1.2 of the main text, Fisher, Kasteleyn, and Temperley designed a method now known as the FKT algorithm for enumerating
the perfect matchings of a planar graph.  This allows us to quickly determine the number of perfect matchings in a dual graph, allowing us to evaluate whether it is computationally tractable to explicitly list all matchings.
FKT  takes a planar embedding of a graph $G$ as input, then assigns signs to the edges in what is called a {\em Pfaffian
orientation} (every face should have an odd number of counterclockwise edges) to create 
a signed, skew-symmetric adjacency matrix $A$. Then $\sqrt{\det A}$ counts the perfect matchings.\footnote{The Pfaffian is a general matrix
operation that  agrees with  $\sqrt{\det A}$  for skew-symmetric matrices.  See
\cite{Loehr} for more mathematical details.}  

This algorithm runs in fractions of a second on each graph, which is fast enough to incorporate at each 
step of a Markov chain. Our implementation is freely available at \cite{Alaska-Repo}, and 
timing details are provided in \S\ref{sec:Enum}.

\section{Prune-and-choose algorithm for constructing perfect matchings}
In order to evaluate the partisan properties of the pairings of House districts, it is not sufficient to count matchings; we also need to generate and examine the full list of matchings. 
In this section, we describe a simple method to create a list of all possible perfect matchings of a graph. This is a recursive method that simplifies the search by looking for forced pairs.

The first step is to prune the graph. This means finding all {\em leaves} of the graph (nodes of degree one, i.e., House districts that are only connected to one other district) and matching each with its only neighbor. 
We call these matches \emph{forced pairs}. One round of pruning may create new nodes of degree one in the resulting graph, and so we iteratively prune forced pairs until there are no nodes of degree one left. 

The second step is a simple  check to rule out a parity obstruction to the existence of a matching.
If any connected component has an odd number of nodes, then it cannot be perfectly matched, 
so the whole graph also fails to have a perfect matching. If all connected components have even numbers of nodes, then we proceed.

Next is a choice step. From the remaining graph, we choose a node of minimum degree,  then consider pairing it with each of its neighbors. For each of those pairings, we remove both nodes from the graph and apply our algorithm to what remains. 
We prune, check, choose, and iterate, in sequence, until the process terminates at a connected graph of two nodes, 
producing  a perfect matching of the original graph. We provide a proof of correctness in Appendix \ref{appendix:validity} and an example run of the algorithm on a small graph below.

\newpage
\subsection{Prune-and-Choose Example}
\label{pruning-alg}

A run of our algorithm on a sample graph with ten nodes proceeds as follows.

\centerline{\begin{tabular}{m{2in}m{2.75in}}
Initiate. & 
\begin{tikzpicture}[scale=.6]
    \node[circle,draw,fill,inner sep=1.5pt, label = right:$J$] (J) at (2.5, 0.5) {};
    \node[circle,draw,fill,inner sep=1.5pt, label = left:$I$] (I) at (2, 0) {};
    \node[circle,draw,fill,inner sep=1.5pt, label = right:$H$] (H) at (0, 0) {};
    \node[circle,draw,fill,inner sep=1.5pt, label = left:$G$] (G) at (-1, 0) {};
    \node[circle,draw,fill,inner sep=1.5pt, label = 30:$F$] (F) at (1.5, 1) {};
    \node[circle,draw,fill,inner sep=1.5pt, label = right:$E$] (E) at (0, 1) {};
    \node[circle,draw,fill,inner sep=1.5pt, label = left:$D$] (D) at (-1, 1) {};
    \node[circle,draw,fill,inner sep=1.5pt, label = right:$C$] (C) at (.75, 1.75) {}; 
    \node[circle,draw,fill,inner sep=1.5pt, label = right:$B$] (B) at (.3, 2.6) {}; 
    \node[circle,draw,fill,inner sep=1.5pt, label = left:$A$] (A) at (-.5, 3) {};                            
    \draw (I) -- (J);
    \draw (F) -- (J);
    \draw (G) -- (H);
    \draw (F) -- (I);
    \draw (E) -- (H);
    \draw (D) -- (H);
    \draw (D) -- (G);
    \draw (D) -- (E);
    \draw (C) -- (F);             
    \draw (E) -- (C);             
    \draw (C) -- (B);                     
    \draw (B) -- (A);  
\end{tikzpicture}\\
\hline
\small  $A$ has degree one, so we record $AB$ as a forced pair. Remove $A,B$ from the graph.
Test for parity. There is one remaining  component with eight nodes, so we pass the parity check. & 
\begin{tikzpicture}[scale=.6]
    \node[circle,draw,fill,inner sep=1.5pt, label = right:$J$] (J) at (2.5, 0.5) {};
    \node[circle,draw,fill,inner sep=1.5pt, label = left:$I$] (I) at (2, 0) {};
    \node[circle,draw,fill,inner sep=1.5pt, label = right:$H$] (H) at (0, 0) {};
    \node[circle,draw,fill,inner sep=1.5pt, label = left:$G$] (G) at (-1, 0) {};
    \node[circle,draw,fill,inner sep=1.5pt, label = 30:$F$] (F) at (1.5, 1) {};
    \node[circle,draw,fill,inner sep=1.5pt, label = right:$E$] (E) at (0, 1) {};
    \node[circle,draw,fill,inner sep=1.5pt, label = left:$D$] (D) at (-1, 1) {};
    \node[circle,draw,fill,inner sep=1.5pt, label = right:$C$] (C) at (.75, 1.75) {}; 
    \node[circle,draw,fill,inner sep=1.5pt, label = right:$B$] (B) at (.3, 2.6) {}; 
    \node[circle,draw,fill,inner sep=1.5pt, label = left:$A$] (A) at (-.5, 3) {};                            
    \draw (I) -- (J);
    \draw (F) -- (J);
    \draw (G) -- (H);
    \draw (F) -- (I);
    \draw (E) -- (H);
    \draw (D) -- (H);
    \draw (D) -- (G);
    \draw (D) -- (E);
    \draw (C) -- (F);             
    \draw (E) -- (C);             
    \draw [gray,dashed] (C) -- (B);                     
    \draw [ultra thick, red] (B) -- (A);       
\end{tikzpicture}\\
\hline
\small
$C$ is lowest-indexed node of degree 2, so try to pair with $E$. However, 
one complementary component ($FIJ$) has an odd number of nodes, so we abandon this branch
of the decision tree.& 
\begin{tikzpicture}[scale=.6]
    \node[circle,draw,fill,inner sep=1.5pt, label = right:$J$] (J) at (2.5, 0.5) {};
    \node[circle,draw,fill,inner sep=1.5pt, label = left:$I$] (I) at (2, 0) {};
    \node[circle,draw,fill,inner sep=1.5pt, label = right:$H$] (H) at (0, 0) {};
    \node[circle,draw,fill,inner sep=1.5pt, label = left:$G$] (G) at (-1, 0) {};
    \node[circle,draw,fill,inner sep=1.5pt, label = 30:$F$] (F) at (1.5, 1) {};
    \node[circle,draw,fill,inner sep=1.5pt, label = right:$E$] (E) at (0, 1) {};
    \node[circle,draw,fill,inner sep=1.5pt, label = left:$D$] (D) at (-1, 1) {};
    \node[circle,draw,fill,inner sep=1.5pt, label = right:$C$] (C) at (.75, 1.75) {}; 
    \node[circle,draw,fill,inner sep=1.5pt, label = right:$B$] (B) at (.3, 2.6) {}; 
    \node[circle,draw,fill,inner sep=1.5pt, label = left:$A$] (A) at (-.5, 3) {};                            
    \draw (I) -- (J);
    \draw (F) -- (J);
    \draw (G) -- (H);
    \draw (F) -- (I);
    \draw (E) -- (H);
    \draw (D) -- (H);
    \draw (D) -- (G);
    \draw (D) -- (E);
    \draw  [gray,dashed]  (C) -- (F);             
    \draw [ultra thick, red] (E) -- (C);             
    \draw [gray,dashed] (C) -- (B);                     
    \draw [ultra thick, red] (B) -- (A);       
\end{tikzpicture}\\
\hline
\small Pairing $C$ instead with $F$ and removing both from the graph forces the $IJ$ pairing. & 
\begin{tikzpicture}[scale=.6]
    \node[circle,draw,fill,inner sep=1.5pt, label = right:$J$] (J) at (2.5, 0.5) {};
    \node[circle,draw,fill,inner sep=1.5pt, label = left:$I$] (I) at (2, 0) {};
    \node[circle,draw,fill,inner sep=1.5pt, label = right:$H$] (H) at (0, 0) {};
    \node[circle,draw,fill,inner sep=1.5pt, label = left:$G$] (G) at (-1, 0) {};
    \node[circle,draw,fill,inner sep=1.5pt, label = 30:$F$] (F) at (1.5, 1) {};
    \node[circle,draw,fill,inner sep=1.5pt, label = right:$E$] (E) at (0, 1) {};
    \node[circle,draw,fill,inner sep=1.5pt, label = left:$D$] (D) at (-1, 1) {};
    \node[circle,draw,fill,inner sep=1.5pt, label = right:$C$] (C) at (.75, 1.75) {}; 
    \node[circle,draw,fill,inner sep=1.5pt, label = right:$B$] (B) at (.3, 2.6) {}; 
    \node[circle,draw,fill,inner sep=1.5pt, label = left:$A$] (A) at (-.5, 3) {};                            
    \draw   [ultra thick, red] (I) -- (J);
    \draw  [gray,dashed] (F) -- (J);
    \draw (G) -- (H);
    \draw  [gray,dashed] (F) -- (I);
    \draw (E) -- (H);
    \draw (D) -- (H);
    \draw (D) -- (G);
    \draw (D) -- (E);
    \draw  [ultra thick, red] (C) -- (F);             
    \draw[gray,dashed] (E) -- (C);             
    \draw [gray,dashed] (C) -- (B);                     
    \draw [ultra thick, red] (B) -- (A);       
\end{tikzpicture}\\
\hline
\small Now $E$ is the lowest-indexed node of minimal degree.  By pairing
 $ED$, the next forced pairing completes the matching.  Similarly for $EH$. & 
\begin{tikzpicture}[scale=.6]
    \node[circle,draw,fill,inner sep=1.5pt, label = right:$J$] (J) at (2.5, 0.5) {};
    \node[circle,draw,fill,inner sep=1.5pt, label = left:$I$] (I) at (2, 0) {};
    \node[circle,draw,fill,inner sep=1.5pt, label = right:$H$] (H) at (0, 0) {};
    \node[circle,draw,fill,inner sep=1.5pt, label = left:$G$] (G) at (-1, 0) {};
    \node[circle,draw,fill,inner sep=1.5pt, label = 30:$F$] (F) at (1.5, 1) {};
    \node[circle,draw,fill,inner sep=1.5pt, label = right:$E$] (E) at (0, 1) {};
    \node[circle,draw,fill,inner sep=1.5pt, label = left:$D$] (D) at (-1, 1) {};
    \node[circle,draw,fill,inner sep=1.5pt, label = right:$C$] (C) at (.75, 1.75) {}; 
    \node[circle,draw,fill,inner sep=1.5pt, label = right:$B$] (B) at (.3, 2.6) {}; 
    \node[circle,draw,fill,inner sep=1.5pt, label = left:$A$] (A) at (-.5, 3) {};                            
    \draw   [ultra thick, blue] (I) -- (J);
    \draw  [gray,dashed] (F) -- (J);
    \draw [ultra thick, blue]  (G) -- (H);
    \draw  [gray,dashed] (F) -- (I);
    \draw [gray,dashed] (E) -- (H);
    \draw  [gray,dashed] (D) -- (H);
    \draw  [gray,dashed] (D) -- (G);
    \draw [ultra thick, blue]  (D) -- (E);
    \draw  [ultra thick, blue] (C) -- (F);             
    \draw[gray,dashed] (E) -- (C);             
    \draw [gray,dashed] (C) -- (B);                     
    \draw [ultra thick, blue] (B) -- (A);       
\end{tikzpicture}
\begin{tikzpicture}[scale=.6]
    \node[circle,draw,fill,inner sep=1.5pt, label = right:$J$] (J) at (2.5, 0.5) {};
    \node[circle,draw,fill,inner sep=1.5pt, label = left:$I$] (I) at (2, 0) {};
    \node[circle,draw,fill,inner sep=1.5pt, label = right:$H$] (H) at (0, 0) {};
    \node[circle,draw,fill,inner sep=1.5pt, label = left:$G$] (G) at (-1, 0) {};
    \node[circle,draw,fill,inner sep=1.5pt, label = 30:$F$] (F) at (1.5, 1) {};
    \node[circle,draw,fill,inner sep=1.5pt, label = right:$E$] (E) at (0, 1) {};
    \node[circle,draw,fill,inner sep=1.5pt, label = left:$D$] (D) at (-1, 1) {};
    \node[circle,draw,fill,inner sep=1.5pt, label = right:$C$] (C) at (.75, 1.75) {}; 
    \node[circle,draw,fill,inner sep=1.5pt, label = right:$B$] (B) at (.3, 2.6) {}; 
    \node[circle,draw,fill,inner sep=1.5pt, label = left:$A$] (A) at (-.5, 3) {};                            
    \draw   [ultra thick, blue] (I) -- (J);
    \draw  [gray,dashed] (F) -- (J);
    \draw  [gray,dashed] (G) -- (H);
    \draw  [gray,dashed] (F) -- (I);
    \draw  [ultra thick, blue]  (E) -- (H);
    \draw [gray,dashed] (D) -- (H);
    \draw   [ultra thick, blue]  (D) -- (G);
    \draw [gray,dashed] (D) -- (E);
    \draw  [ultra thick, blue] (C) -- (F);             
    \draw[gray,dashed] (E) -- (C);             
    \draw [gray,dashed] (C) -- (B);                     
    \draw [ultra thick, blue] (B) -- (A);       
\end{tikzpicture}
\end{tabular}}

\noindent In the end, we find the two perfect matchings $AB/CF/IJ/ED/GH$ and $AB/CF/IJ/EH/GD$.

\subsection{Prune-and-choose algorithm validity}
\label{appendix:validity}

In this section we formally describe the prune-and-choose method and provide a proof of correctness. Pseudo-code for the algorithm is 
given here  and our implementation in Python is available at \cite{Alaska-Repo}. We introduce some additional notation to describe the method. The subgraph of $G$ induced by deleting nodes $u$ and $v$ will be denoted $G\setminus\{u,v\}$. We will represent a matching as a set of edges $M=\{(u_1,v_1), (u_2,v_2),\ldots, (u_\ell,v_\ell)\}$. We assume that the vertices of $G$ are ordered in order to provide a deterministic algorithm.  To generate the full set of matchings for a graph $G$, we would call $\textsc{FindMatchings}(G,\emptyset)$.


\begin{algorithm}
\caption{Pseudo-code for Prune-and-Choose Algorithm to Find All Perfect Matchings in a Graph  $G$}\label{alg:PM}
\begin{algorithmic}[1]
\Procedure{FindMatchings}{$G, M$}\Comment{Input a graph $G$ and the current set of matched edges $M$}
\If{$G$ is connected and has exactly two vertices $u, v$}
\State  \textbf{return} $G\setminus \{u,v\}$, $M \cup (u,v)$ 
\ElsIf{G has any vertex with exactly one neighbor}
\State \textbf{prune:} let $u$ be the first degree-one vertex; let $v$ be its neighbor
\State \textbf{return} \Call{FindMatchings}{$G\setminus\{u,v\},M\cup(u,v)$} \Comment{Pair forced vertices and recurse. } 
\ElsIf{$G$ contains a component with an odd number of vertices}
\State \textbf{break} \Comment{There are no perfect matchings in $G$}
\Else
\State let $u$ be first vertex with a minimum number of neighbors $v_1, ..., v_k$
\State for $1\leq i\leq k,$ let $G_i = G\setminus \{u,v_i\}$ and $M_i = M \cup (u,v_i)$
\State \textbf{return} $\bigcup_{i=1}^k$\Call{FindMatchings}{$G_i$, $M_i$}  \Comment{Recurse to find all perfect matchings with each pair}
\EndIf
\EndProcedure
\end{algorithmic}
\end{algorithm}

 We next show that the algorithm returns the correct set of perfect matchings on any graph. They key idea of the algorithm and the proof is that for any edge of the graph, the set of perfect matchings that contain edge $(u,v)$ can be computed by finding all perfect matchings in the subgraph $G\setminus\{u,v\}$. This is an example of the self--reducible nature of the perfect matching problem which is discussed in more detail below. 
 
\begin{theorem}
The prune-and-choose algorithm correctly finds all perfect matchings in the input graph.
\end{theorem}

\begin{proof}


We consider any graph $G$ with $n=2k$ vertices and proceed by induction on $k$.  When $k=1,$ $G$ is either connected (in which case the algorithm correctly finds the unique perfect matching at lines 2-3) or has two isolated vertices and no perfect matchings (which the algorithm correctly reports in lines 7-8).  

For $k>1$ the algorithm proceeds according to exactly one of the following three cases:
\begin{enumerate}
    \item If $G$ contains a leaf $u$ with neighbor $v$, then $u$ must be matched to $v$ in any perfect matching. Line 6 then calls \textsc{FindMatchings}  on $G\setminus\{u,v\}$  which returns the correct set of matchings of $G\setminus\{u,v\}$,  by our inductive hypothesis. Adding $(u,v)$ to each matching returned by this function gives the full set of matchings for $G$.
    \item If $G$ contains no leaves and some connected component of $G$ has an odd number of vertices, then there are no perfect matchings in $G$ and the algorithm correctly terminates at lines 7-8.
    \item If $G$ contains no leaves and each connected component of $G$ has an even number of vertices, then there exists a vertex of minimal index $u$ which has a minimum number of neighbors. Since $G$ has no leaves and no odd components, $u$ has degree at least $2$.  In any perfect matching, it must be matched to one of its neighbors $v_1, ..., v_k$.  The algorithm considers each possibility calling \textsc{FindMatchings}  on $G\setminus\{u,v_i\}$ which returns the correct set of matchings by our inductive hypothesis. As in step 1, adding $(u,v_i)$ to the matchings returned on $G\setminus\{u,v_i\}$ provides a complete set of matchings for $G$. 
\end{enumerate}   
Thus for any graph $G$, the first pass through the algorithm either returns $\emptyset$, which only occurs if $G$ has no perfect matchings, or it calls the algorithm recursively on a graph of size $2(k-1)$. These recursive calls satisfy our inductive hypothesis and hence we obtain the complete set of matchings for $G$. 
\end{proof}

We note that well-known classes of planar graphs have exponentially many perfect matchings.  For example, this is true of the $n\times n$ grids  \cite{Bjo10, kasteleyn_statistics_1961,temperley_dimer_1961}. This trivially implies that there is no polynomial-time algorithm to list them all as output. As we discuss in Appendix \ref{sec:Enum}, the dual graphs of real-world districting plans often have more perfect matchings than a grid graph of comparable size. In that section we also provide timing results for our algorithm that demonstrate that it is adequately fast for several problems
at realistic scale, but not all. In Appendix~\ref{appendix:sampling} below we show how a sampling approach can be employed to those settings in which listing all perfect matchings is computationally infeasible. 


\section{Enumerating matchings}\label{sec:Enum}

Here we apply FKT and Prune-and-Choose to compute the number of potential matchings for each of the eight states that require Senate districts to be formed from adjacent pairs of House districts.  We use the standard Census shapefiles
to generate dual graphs for each state, making them parallel to the 
permissive graph for Alaska.  Visualizations of these dual graphs are shown in Appendix A.

 \begin{table}[ht]
 \centering
 \begin{tabular}{|r|c|c|c|c|}
 \hline
 &\small Alaska&\small Illinois&\small Iowa&\small Minnesota\\
 \hline
 \hline
\small House districts&40 &118 &100 &134 \\
 \hline
 \small Dual edges&100&326&251&260\\
 \hline
 \small Matchings&\footnotesize 108,765&{\footnotesize 9,380,573,911}&{\footnotesize 1,494,354,140,511}&
{\footnotesize  6,156,723,718,225,577,984}\\
 \hline
 \small FKT runtime&0.027 sec &0.39 sec&0.21 sec&0.53 sec\\
 \hline
\end{tabular}

\medskip

 \begin{tabular}{|r|c|c|c|c|}
 \hline
 &\small Montana&\small Nevada&\small Oregon&\small Wyoming\\
 \hline
 \hline
 \small House districts&100&42&60 &60 \\
 \hline
 \small Dual edges&269&111&158&143\\
 \hline
 \small Matchings&{\footnotesize 11,629,786,967,358}&
 \footnotesize 313,698&\footnotesize 229,968,613&
 \footnotesize 920,864\\
 \hline
 \small FKT runtime&0.24 sec&0.038 sec&0.079 sec&0.056 sec\\
 \hline
 \end{tabular}
 \caption{Number of matchings possible with 
 respect to the current House plan for each state (with dual graphs generated from census shapefiles) and timings for computing them with FKT.\label{tab:timingdata}}
 \end{table}

Alaska's 108,765 (permissive) matchings are the
fewest among the eight states. 
This is partially due to the fact that Alaska has fewer House districts than the other states, and partly due to 
lower edge density and several forced matchings.%
\footnote{The four districts in the Southeast corner of the state must be paired  (33--34 and 35--36), further restricting the possible matchings.} 
The number of matchings varies
greatly across the other states, with Minnesota having the most 
at  $6.1\times 10^{18}$, over six quintillion. 
By contrast, the number of matchings in a
$10\times 10$ grid with 100 nodes is 258,584,046,368 (fewer
than Iowa, which has 100 House districts) and 
a $12\times 12$ grid with 144 nodes has 
53,060,477,521,960,000 (fewer than Minnesota, which has 134 House districts). 
To understand why the states' matching numbers exceed those of comparably sized grids, 
consider the impact of just a few extra edges.  
Adding just four edges to the $10\times 10$ grid---a single diagonal edge from each of the four corner vertices to its diagonal neighbor---increases the number of matchings by 745,241,088.


Extrapolating the prune-and-choose timing from Nevada (299 seconds) and Wyoming (851 seconds) suggests that generating all of the matchings for some of the other states would take prohibitively long---even with linear scaling, the Wyoming timing
suggests that the Minnesota computation would take some 180 million years.
However, it is possible to sample matchings from planar graphs uniformly, allowing
for good estimates of relevant statistics. We have implemented the technique suggested in \cite{Jerrum} for this purpose \cite{Alaska-Repo} and in Appendix~\ref{appendix:sampling} we validate  this approach 
on Alaska.

\section{Sampling and extremization over
matchings}
\label{appendix:sampling}

For Minnesota's 6.1 quintillion matchings, it would be 
prohibitively inefficient to list them all, no matter the algorithmic design. On the other hand, we can construct uniform samples of the full set of matchings by making use of the {\em self-reducible structure} in the perfect matching problem \cite{Jerrum} as follows.
We can compute the likelihood that a given edge appears in a perfect matching by deleting the edge from the graph and enumerating the matchings on the remaining nodes with FKT. The ratio of matchings
on the leftover to total matchings is the probability that the edge
is used.  With this, we can iterate, starting with the original graph and adding a single edge to the matching at each step with appropriate probability. Since FKT  runs in polynomial time, so does our sampling procedure, since a perfect matching requires $\frac{n}{2}$ edges and finding the probabilities used to select each edge requires at most $\binom{n}{2}$ FKT evaluations.

We next demonstrate that the uniform sampling method can attain good accuracy with a reasonably small number of samples, using the case of Alaska where we can compare to the ground truth from the full matching set. For each of our three dual graphs, we sample 100 matchings uniformly and compare the resulting statistics to those of the full set of matchings. Figure \ref{fig:sample} shows these comparisons. Although the distributions are not identical, they are quite similar and the sample means vary only by small fractions of a seat from the actual values. 

\begin{figure}[!h]
    \centering
\begin{tikzpicture}[scale=.7]
\node at (0,-3) {\small Cong18-A D Senate seats} ;
\node at (8,-3) {\small Gov 18-A D Senate seats};
\node at (0,0) {\includegraphics[width=2in]{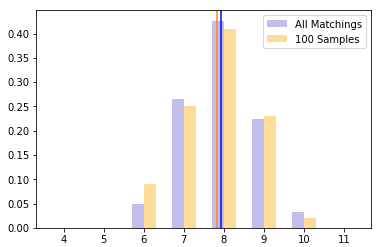}};
\node at (8,0) {\includegraphics[width=2in]{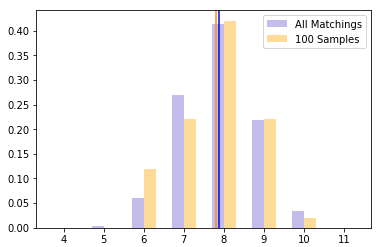}};
\end{tikzpicture}

\vspace{.2in}

\begin{tabular}{|l|c|c|c|}
    \hline
      abs. error   &  Tight&Restricted&Permissive\\
      \hline
      \hline
    Cong18-A &0.0082&0.0234& 0.0868\\
    \hline
    Gov18-A & 0.0203&0.0361&0.0814\\
    \hline
    \end{tabular} 
    \caption{Comparison of number of Democratic Senate districts in a uniform sample of 100 (permissive) matchings to the full collection of matchings. The table shows the absolute error in the average seats total for this and the other
    two Alaska dual graphs.  The histograms show more detail, and illustrate how close the averages are with only 1/1000 of the space being sampled.  \label{fig:sample}}
\end{figure}

This example shows that even  a sample of modest size produces a good estimate of the full distribution. This provides support for our assertion that this procedure can be carried out successfully on states like Minnesota, where it would be computationally infeasible to generate all matchings.  We note that all materials are available in our code repositories for others
to perform this sampling for the other matching states, but there will be a non-trivial data setup cost in choosing appropriate
election data and cleaning it for the analysis.

Though the histograms above are quite similar, the 
sample fails to capture the full range of 
seat outcomes in the Governor's race:  a small number of possible matchings result in 
five D seats, but that is never observed in the sample.
A second algorithm may be employed to provably find the correct 
range of seats outcomes possible, again without fully listing 
the matchings.  Finding perfect matchings of extremal weight, given an edge-weighted graph, is a classic problem in combinatorics, solved for instance with the Blossom algorithm
developed by Edmonds in the 1960s \cite{Ed65a, Ed65b}.
To apply that in this setting, we use any given pattern of votes to assign a 
weight to each edge of our dual graph: an edge $\{u, v\}$ linking two House districts $u$ and $v$ is given weight 1 if there are more D than R votes in the hypothetical Senate district that combines $u$ and $v$. Otherwise, assign  weight 0.  The weight of the perfect matching is defined as the sum of the weights of its edges.  By 
construction, this is the number of D seats in that matching.
For more background on extremal perfect matchings, see for instance Chapters 25-26 of \cite{Sch03}.


As a final note, knowing these extremes also informs the size of a uniform sample necessary to estimate the true distribution to a desired precision. A detailed discussion of the precise number of samples needed for various estimates is presented in \cite{bo}. In particular, Theorem 5.3 shows that with failure rate $\delta$, taking $\max\left(\frac{4}{\varepsilon^2},\frac{4\ln(\frac{1}{\delta})}{\varepsilon^2}\right)$ samples suffices to estimate the probability of each individual outcome to within $\varepsilon$
(i.e., an $L^\infty$ bound) whereas $\max\left(\frac{4n}{\varepsilon^2},\frac{8\ln(\frac{1}{\delta})}{\varepsilon^2}\right)$ samples suffice to bound the sum of the absolute differences between the individual estimates and the true probabilities (an $L^1$ bound).


\bibliographystyle{plain}

\bibliography{main}

\end{document}